\documentclass{article}

\usepackage{amsmath}
\usepackage{amssymb}
\usepackage{amsthm}
\usepackage{color}
\usepackage{url}
\usepackage{graphicx}

\newtheorem{theorem}{Theorem}[section]
\newtheorem{proposition}[theorem]{Proposition}
\newtheorem{lemma}[theorem]{Lemma}

\newcommand{\CC}{\mathbb{C}}

\newcommand{\sF}{\mathcal{F}}
\newcommand{\sL}{\mathcal{L}}
\newcommand{\sP}{\mathcal{P}}

\newcommand{\ZZ}{\mathbb{Z}}
\newcommand{\RR}{\mathbb{R}}
\newcommand{\DD}{\mathbb{D}}
\newcommand{\HH}{\mathbb{H}}

\newcommand{\sinc}{\text{sinc}}
\newcommand{\sgn}{\text{sgn}}
\newcommand{\loc}{\text{loc}}
\renewcommand{\Im}{\mathrm{Im}}

\def\s{{{x}}}
\def\z{{{u}}}
\def\Z{{{W}}}
\def\m{{{m}}}
\def\ts{{\widetilde\s}}
\let\vv\v
\let\vold=\v
\def\v{{{s}}}
\def\u{{{y}}}
\def\xx{{{v}}}
\def\k{{{k}}}
\newcommand{\kk}{K}
\def\XX{{F^\nu_0(\RR)}}
\def\Cm{C^{{m}}}
\def\Const{{C\,}}
\def\epsilon{{\varepsilon}}
\def\si{{\sigma}}
\def\a{{w_{0,\beta,\infty}^\pm}}
\def\b{{d_a^\pm}}

\begin{document}

\title{Travelling heteroclinic waves in a Frenkel-Kontorova chain with anharmonic on-site potential} \author{Boris Buffoni,
  Hartmut Schwetlick and Johannes Zimmer\thanks{University of Bath, Department of Mathematical Sciences, Bath BA2 7AY, United Kingdom, \texttt{zimmer@maths.bath.ac.uk}}}

\date{}

\maketitle

\begin{abstract}
 The Frenkel-Kontorova model for dislocation dynamics from 1938 is given by a chain of atoms, where neighbouring atoms interact
  through a linear spring and are exposed to a smooth periodic on-site potential. A dislocation moving with constant speed
  corresponds to a heteroclinic travelling wave, making a transition from one well of the on-site potential to another. The
  ensuing system is nonlocal, nonlinear and nonconvex. We present an existence result for a class of smooth nonconvex on-site
  potentials. Previous results in mathematics and mechanics have been limited to on-site potentials with harmonic wells. To
  overcome this restriction, we propose a novel approach: we first develop a global centre manifold theory for anharmonic wave
  trains, then parametrise the centre manifold to obtain asymptotically correct approximations to the solution sought, and finally
  obtain the heteroclinic wave via a fixed point argument.  

\end{abstract}

\section{Introduction}
\label{section: particular} 

In 1938, Frenkel and Kontorova~\cite{Frenkel1939a} proposed a model for plastic deformations and twinning, given by an infinite
chain of nonlinear oscillators linearly coupled to their nearest neighbours,
\begin{equation}
  \label{eq:fk-orig-IVP}
   \ddot \upsilon_j(t) = \gamma\left[ (\upsilon_{j+1}(t) - \upsilon_j(t)) - (\upsilon_j(t) - \upsilon_{j-1}(t))\right] - 
  g'(\upsilon_j(t)).
\end{equation}
These are Newton's equation of motion for atom $j \in \ZZ$ with mass $1$; $\gamma$ is the elastic modulus of the elastics springs
and $g$ is smooth and periodic.

Travelling waves are particularly simple forms of coherent motion; here they are of the form $\upsilon_j(t)=u(j-ct)$ with some
travelling wave profile $u$. Equation~\eqref{eq:fk-orig-IVP} written in travelling wave coordinates $\s := j -ct$, with $c$ being
the wave speed, is
\begin{equation}
  \label{eq:fk-orig}
  c^2u''(\s)- \gamma \Delta_D u(\s) + g'(u(\s)) =0,
\end{equation}
where $\Delta_D$ is the discrete Laplacian
\begin{equation}
  \label{eq:lap}
  (\Delta_D u)(\s):=u(\s+1)-2u(\s)+u(\s-1).
\end{equation}
In the original paper~\cite{Frenkel1939a}, the force of the on-site potential is (in suitable units) $g'(u) =\sin(2 \pi u)$.

Equation~\eqref{eq:fk-orig} is an advance-delay differential-difference equation of Hamiltonian nature, nonlocal and
nonlinear. Proving the existence of small solutions to~\eqref{eq:fk-orig} has been a major challenge, accomplished only in 2000 in
the seminal paper by Iooss and Kirchg\"assner~\cite{Iooss2000a}. They establish the existence of \emph{small} amplitude solutions,
under the convexity assumption $g''(0) > 0$.  In particular, Iooss and Kirchg\"assner prove the existence of nanopterons, that is,
localised waves which are superimposed to a periodic (``phonon'') wave train. Another remarkable result is the existence of
breathers (spatially localised time-periodic solutions) by MacKay and Aubry~\cite{MacKay1994a}. There is a wealth of studies of
Frenkel-Kontorova models. We refer the reader to the monograph by Braun~\cite{Braun1998a} and only mention more recent results for
sliding states by Qin for a forced Frenkel-Kontorova chain, both with and without damping~\cite{Qin2010a,Qin2012a} and periodic
travelling waves (wave trains) by Fe{\vv{c}}kan and Rothos~\cite{Feckan2007a}.

The mathematical theory of existence of travelling wave dislocation as originally posed in~\cite{Frenkel1939a}, however, is still
largely open. One reason is that dislocations are large solutions, making the transition from one well of $g$ to another, and
therefore experience the nonconvexity of on-site potential. We highlight a few results for the analysis of travelling dislocations
for the chain~\eqref{eq:fk-orig}.  An early study is that of Frank and van der Merwe~\cite{Frank1950a}, where the continuum
approximation of~\eqref{eq:fk-orig}, the sine-Gordon equation, is analysed. Rigorously, the dangers of relying on the PDE
counterpart of a lattice equation were realised decades later (though Schr\"odinger pointed out this difference in his ingenious
analysis~\cite{Schrodinger1914a}). In particular, Iooss and Kirchg\"assner~\cite{Iooss2000a} prove the existence of infinitely
many types of travelling waves which do not persist in the continuum approximation. Friesecke and Wattis~\cite{Friesecke1994a}
study the Fermi-Pasta-Ulam-Tsingou (FPUT) chain (nonlinear interaction between nearest neighbour atoms and $g\equiv0$) and obtain
the remarkable result that in a spatially discrete setting, solitary waves exist quite generically, not just for integrable
systems (such as the so-called Toda lattice). Recently, explicit solitary waves have also been constructed for the FPUT chain with
piecewise quadratic potential~\cite{Truskinovsky2014a}.

The analysis of dislocation solutions to the lattice equation~\eqref{eq:fk-orig} relies in previous work on the assumption that
$g$ is piecewise quadratic; then the force $g'$ in~\eqref{eq:fk-orig} is piecewise linear and Fourier methods can be applied. We
refer the reader to Atkinson and Cabrera~\cite{Atkinson1965a} (note that some findings of that paper have been questioned in the
literature~\cite{Earmme1977a}) and extensive work by Truskinovsky and collaborators, both for the Fermi-Pasta-Ulam-Tsingou chain
with piecewise quadratic interaction~\cite{Truskinovsky2005a} and the Frenkel-Kontorova model~\cite{Kresse2003a}. Kresse and
Truskinovsky have also studied the case of an on-site potential with different moduli (second derivatives at the
minima)~\cite{Kresse2004a}. Slepyan has made a number of important contributions, for
example~\cite{Slepyan2005a,Slepyan2005b}. Flytzanis, Crowley, and Celli~\cite{Flytzanis1977a} apply Fourier techniques to a
problem where the potential consists of three parabolas, the middle one being concave.

To the best of our knowledge, the original problem of a dislocations exposed to an anharmonic on-site potential has only been
amenable to careful numerical investigation~\cite{Savin2000a}. Obviously, mathematically, the use of Fourier tools as for the
results discussed in the previous paragraph is no longer possible. Physically, the introduction of such a nonlinearity changes the
nature of the system fundamentally, as modes can now mix. The physical interpretation of the result presented in this paper is
that despite this change, solutions exist and, remarkably, can be obtained via a perturbation argument from the degenerate case of
piecewise quadratic wells (the degeneracy manifests itself in a twofold way, firstly in the prevention of mode mixing and secondly
in a singularity of the force $g'$ at the dislocation line; the existence result presented here holds for the physically realistic
case of smooth forces). We develop what seems to be a novel approach to prove existence for systems with small nonlinearities (in
our case $g$ being anharmonic but close to a piecewise harmonic potential; the reason for this restriction is that we apply a
fixed point argument). We first obtain a detailed understanding of wave trains in the anharmonic (but near harmonic) wells of the
on-site potential; this is obtained by a global centre manifold description, much in the spirit of the local analysis of Iooss and
Kirchg\"assner~\cite{Iooss2000a}. Unlike them, we do not perform a normal form analysis but instead construct a parametrisation of
the centre manifold. From this knowledge, it is possible to construct a one-parameter family $w_\beta$, $\beta\in[-1,1]$, of
approximate (asymptotically correct, as $\s \to \pm \infty$) heteroclinic solutions of the Frenkel-Kontorova travelling wave
equation.  This step can be seen as a homotopy method from solutions or approximate solutions to the problem with piecewise
quadratic wells (the homotopy parameter being $\epsilon\geq 0$ in Theorem \ref{theo:main}, even if in the end we do not rely on
continuity with respect to $\epsilon$, but rather on the smallness of $\epsilon>0$).  This is an unconventional step in the sense
that we do not attempt to find a homotopy between solutions to a family of problems, but only between approximate solutions. In a
final step, the heteroclinic travelling wave solution is obtained from the approximate solutions via a topological fixed point
argument. A key property is that the family $w_\beta$ satisfies a transversality condition with respect to $\beta$ (see~\eqref{eq:
  K_0}).  This method, in a much simpler setting in which centre manifold theory is not used, was developed
in~\cite{Buffoni2017a}.

On an abstract level, the approach developed here allows for the passage from a linear problem to a moderately nonlinear one. We
remark that the analysis of the linear problem (\eqref{eq:fk-orig} with piecewise quadratic on-site potential) is challenging in
its own right, and has been solved mathematically by a de-singularisation of the Fourier image of the
solution~\cite{Schwetlick2009a,Kreiner2011a}. While the detailed arguments we give below are admittedly rather technical, the
method developed here might also be useful for the numerical computation of solutions to such nonlinear problems; indeed, the
centre manifold approach can guide the implementation of a path-following technique, while the fixed point argument can for
example translate into a gradient descent method.

We have chosen the original Frenkel-Kontorova equation~\eqref{eq:fk-orig-IVP} but remark that some of the references given above
study a modified model, with an added force. There are also extensions to higher space dimensions, for
example~\cite{Srolovitz1986a}. The methodology of this paper should in principle apply to these problems as well.

The result proved here covers cases of~\eqref{eq:fk-orig} with $g'$ anharmonic, periodic and $C^\infty$. As our argument is
perturbative in nature, it is not clear whether the particular choice of a trigonometric potential made in~\cite{Frenkel1939a} is
covered; we have no explicit control over the range of perturbations covered. However, while the choice of a trigonometric
function as made by Frenkel and Kontorova~\cite{Frenkel1939a} is natural, there is no intrinsic reason to prefer such an on-site
potential. Here, we make for simplicity the choice $\gamma=1$ and place two neighbouring minima of $g$ at $\pm 1$.

At the end of Section~\ref{sec:Setting-main-result}, we give a plan of the paper, summarising the required steps and linking them
to the relevant sections. Throughout the paper, $C$ is a constant that may change from line to line; however, $C$ is independent
of $u$ and of small enough $\epsilon$.

\section{Setting and main result}
\label{sec:Setting-main-result}

We can assume that in travelling wave coordinates, the dislocation line is at the origin $\s=0$; then all atoms in the left
half-line are in one well of the on-site potential $g$ and all atoms in the right half-plane are in the neighbouring well on the
right.  It is no loss of generality to consider an on-site potential $g$ with two wells, rather than a periodic one.  Indeed, the
solutions we obtain for a two-well potential are also solutions for the same equation with a periodic potential.  This is implied
by our approach to obtain a special two-well solution as a sum of an associated particular solution and a corrector, both being
uniformly bounded. Since upper and lower bounds on the solution are available, the solution will also solve the problem for a
periodic potential built by extension from the two-well potential.  We thus show the existence of heteroclinic waves for
\begin{equation}
  \label{eq:main}
  c^2u''(\s)-\Delta_D u(\s)+\alpha u(\s)-\alpha \psi'(u(\s)) =0, 
\end{equation}
where $\psi'$ is a perturbation of the sign function. We note that~\eqref{eq:main} is~\eqref{eq:fk-orig} with the choice
\begin{equation}
  \label{eq:g}
  g(u) := \frac \alpha 2 (u^2+1) - \alpha \psi(u). 
\end{equation}
This choice is made since $g$ is a perturbation of the piecewise quadratic on-site potential 
\begin{equation*}
  g_0 (u) := \frac \alpha 2 \min\left((u+1)^2,(u-1)^2\right)   
\end{equation*}
(see Fig.~\ref{fig:g}); indeed, the force associated with $g_0$ is $g_0'(u) = \alpha u - \alpha \sgn (u)$, while the force
associated with $g$ is $g'(u) = \alpha u - \alpha \psi' (u)$.  We point out that $g$ (and hence $g_0$, with the choice $\psi'(u) =
\sgn(u)$) is a double-well potential, mimicking two neighbouring wells of the trigonometric potential proposed by Frenkel and
Kontorova. Precise assumptions on $\psi$ are stated in Theorem~\ref{theo:main} below.
\begin{figure}
  \begin{center}
    \includegraphics[scale=0.5]{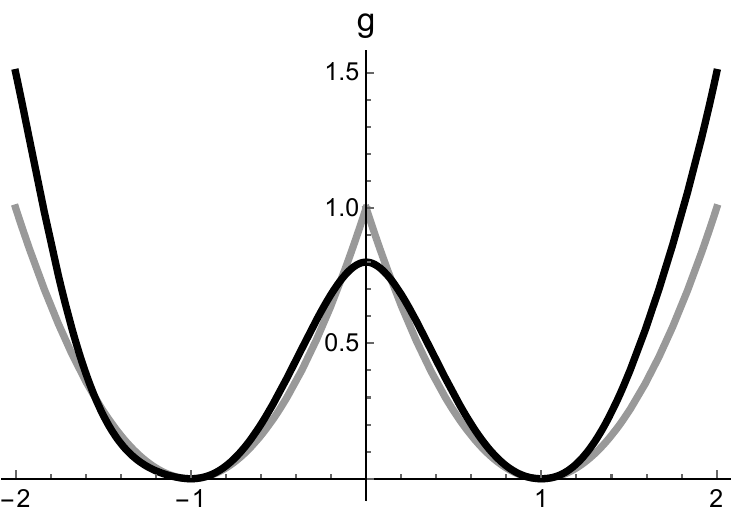}
  \end{center}
  \caption{Schematic plot of an anharmonic on-site potential $g$ as studied in this article (black); the piecewise harmonic
    potential $g_0$ is plotted for comparison (grey). The results presented in this article also hold for smooth periodic
    extensions of potentials $g$ as shown here.}
  \label{fig:g}
\end{figure}

To motivate some assumptions in the main theorem, we briefly inspect the linear part of~\eqref{eq:main},
\begin{equation}
  \label{eq:L}
  u\rightarrow Lu:=c^2u''-\Delta_D u+\alpha u. 
\end{equation}
In Fourier space, $L$ is written as 
\begin{equation*}
  -c^2\k^2+2(1-\cos \k)+\alpha=
  -c^2\k^2+4\sin^2(\k/2)+\alpha
  =:D(\k), 
\end{equation*}
where $D$ is the \emph{dispersion function}. Let $\alpha$ be given by
\begin{equation}
  \label{eq:alpha}
  \alpha := c^2\left(\frac\pi 2\right)^2 - 2; 
\end{equation}
this choice was also made in~\cite{Kreiner2011a}. Then trivially
\begin{equation}
  \label{eq:k0}
  \k_0 := \frac \pi 2
\end{equation}
is one root of $D$ and $-\k_0$ is another. Furthermore, for $c=1$, $D'(\k) =-2c^2\k+2\sin\k$ vanishes only at $\k=0$. Thus, if $c$
is sufficiently close to $1$ (we will only consider the case where additionally $c \leq 1$), then $D$ vanishes exactly at $\k_0$
and $-\k_0$. This is the key property of $D$ used in this paper.

The main result of the paper is the following theorem.
\begin{theorem}
  \label{theo:main}
  We consider the equation~\eqref{eq:main},  
  \begin{equation*}
    c^2u''-\Delta_D u+\alpha u-\alpha \psi'(u) =0
  \end{equation*}
  on $\RR$, where $\Delta_D$ is the discrete Laplacian defined in~\eqref{eq:lap}. For small $\epsilon\in(0,1/2)$, the
  on-site potential $\psi_\epsilon$ is assumed to be an even function $\psi=\psi_\epsilon\in C^\infty(\RR,\RR)$
  satisfying the following conditions. Let
  \begin{equation}
    \label{eq:psi2prime}
    |\psi_\epsilon''(u)|\leq 2\epsilon^{-1} \text{ for } |u| <\epsilon,    
  \end{equation}
  and, for $|u|\geq \epsilon$,
  \begin{equation}
    \label{eq:psiprimeoutside}
    |\psi_\epsilon'(u)-\sgn(u)|< \Const \epsilon,
  \end{equation}
  and, again for $|u|\geq \epsilon$,
  \begin{equation}
    \label{eq: cdn on psi} 
    |\psi_\epsilon''(u)|<\Const \epsilon,~~
    |\psi_\epsilon'''(u)|<\Const \epsilon,~~
    |\psi_\epsilon^{(4)}(u)|<\Const \epsilon,~~
    |\psi_\epsilon^{(5)}(u)|<\Const  \epsilon
  \end{equation}
  (there is no condition on $\psi'''(u)$ for $|u|<\epsilon$).  
  We also assume that $\psi'_\epsilon(u)-u$ vanishes at $u=1$ and $u=-1$ (for small $\epsilon>0$).

  Let $\k_0$ be given by~\eqref{eq:k0} and $\alpha$ be given by~\eqref{eq:alpha}. If $\epsilon>0$ is small enough, then there
  exists a range of velocities $c\leq 1$ close to $1$ such that for these velocities, there exists a heteroclinic solution
  to~\eqref{eq:main}. Here heteroclinic means that the asymptotic state near $-\infty$ is in one well of the on-site potential
  while the state near $+\infty$ is in the other.
\end{theorem}
The assumption that $\psi'_\epsilon(u)-u$ vanishes at $u=1$ and $u=-1$ could be relaxed; it has been added to simplify the
arguments. Physically, this assumption means that $u=1$ and $u=-1$ are equilibria (that is, constant solutions
to~\eqref{eq:main}).

The proof of Theorem \ref{theo:main} is given in Section~\ref{sec:Exist-heter-conn}, using results of
Sections~\ref{sec:Construction-asymp} and~\ref{sec:Prop-family-wbeta}. Since the proof is convoluted and technical, we give here
an outline. To formulate the sequence of steps, we first introduce some notation. We begin by defining exponentially weighted
function spaces as in~\cite{Iooss2000a}.  For $\nu\in \RR$, $\m\in\{0,1,2,\ldots\}$ and a Banach space $X$, we denote by
$E^\nu_\m(X)$ the Banach space of functions $f\in \Cm(\RR,X)$ such that
\begin{equation} 
\label{eq:Eknu}
  ||f||_{E^\nu_\m(X)}:=\max_{0\leq j\leq \m}
  ||e^{-\nu |\cdot|} f^{(j)}||_{L^\infty(\RR,X)}<\infty.
\end{equation}

For $X=\RR$, continuous functions which decay exponentially at $\pm \infty$ are contained in the spaces $E^\nu_0(X)$, for some
negative $\nu<0$.  We also require analogous function spaces where the $\m$th derivative is not continuous, but only in
$L^\infty_{loc}(\RR)$. So let $F^\nu_\m(X)$ the Banach space of functions $f\in W^{m,\infty}(\RR,X)$ such that
\begin{equation} 
\label{eq:Fknu}
  ||f||_{F^\nu_\m(X)}:=\max_{0\leq j\leq \m}
  ||e^{-\nu |\cdot|} f^{(j)}||_{L^\infty(\RR,X)}<\infty.
\end{equation}

If, in the definitions above, the function $f$ is only required to be defined on an open subset $A\subset \RR$, we shall write
$E^\nu_m(A,X)$ and $F^\nu_m(A,X)$ respectively, where $L^\infty(\RR,X)$ is replaced by $L^\infty(A,X)$ in these definitions.
 
\emph{Step 1: Special (degenerate) case, $\epsilon=0$.} In the limit case $\epsilon=0$, $\psi$ is not smooth at $0$
by~\eqref{eq:psiprimeoutside}, as $\psi'(\s) = \sgn(\s)$; the choice of $\psi(0)$ is immaterial. Also, $\psi$ satisfies
$\psi'(\pm 1)=\pm 1$ and $\psi''(u)=0$ on $(-\infty,0)$ and on $(0,\infty)$.  We use an existence result~\cite{Kreiner2011a} for
heteroclinic odd solutions $u_p\in H^2_{\loc}(\RR)$ for the special case $\psi' = \sgn$ in~\eqref{eq:main},
\begin{equation}
  \label{eq:kz}
  c^2u''-\Delta_D u+\alpha u-\alpha \text{sgn}(u)=0
\end{equation}
on $\RR$. The parameters $\alpha, k_0$ and $c$ are as in Theorem~\ref{theo:main}.

For $|\lambda|<1$ and $\theta\in[0,2\pi)$, trivially $1+\lambda\sin(\k_0\s+\theta)$ is a solution to~\eqref{eq:kz} on $[1,\infty)$
and $-1+\lambda\sin(\k_0\s-\theta)$ is a solution on $(-\infty,-1]$. From~\cite{Kreiner2011a}, we will use that there exists a
function $u_p$ that solves~\eqref{eq:kz} and satisfies
\begin{equation}
  \label{eq:jodel}
  \lim_{\s\rightarrow \pm\infty}\left(u_p(\s)\mp 1-\lambda\sin(k_0\s\pm \theta)\right)=0
\end{equation}
for some $\lambda$ and $\theta$. 

The core argument for the next Steps 2--4 is to build a particular family of approximate solutions
$w_\beta \in W^{2,\infty}(\RR)$, which are $C^1$ as a function of $\beta\in[-1,1]$, and asymptotically, as $\s \to \pm \infty$,
they approximate a heteroclinic travelling wave solution. However, they are allowed to be far from a solution near the
dislocation, $\s = 0$.  Step 2 provides such a family for $\epsilon =0$, Step 3 extends this existence result for the case
$\epsilon >0$ we are interested in, and Step 4 uses this family of approximate solutions to obtain an exact solution.

\emph{Step 2: ``Almost solution'' family for the special (degenerate) case, $\epsilon=0$.}  We will to construct a particular
family of functions $[-1,1]\ni \beta\rightarrow w_{0,\beta}\in W^{2,\infty}(\RR) \cap C^2(\RR\backslash\{0\})$, which are odd in
$\s$ and such that $\beta\rightarrow w_{0,\beta}$ is $C^1$ in $\beta$. Further, the $w_{0,\beta}$ asymptotically, as
$\s \to \pm \infty$, converge to a heteroclinic travelling wave solution.  We do not require them to be close to a solution near
the dislocation, $\s = 0$.

For $\epsilon=0$, such a family $w_{0,\beta}$ is obtained by choosing
\begin{equation}
  \label{eq: def of wob}
  w_{0,\beta} =u_p+B\beta u_o\, 
\end{equation}
for some small constant $B>0$, where $u_p\in W^{2,\infty}(\RR) \cap C^2(\RR\backslash\{0\})$ is the particular odd solution
to~\eqref{eq:kz} of~\cite{Kreiner2011a} discussed in Step 1; the odd function $u_o\in C^{4}(\RR)$ vanishes in a neighbourhood of
$0$ and satisfies for some $\nu<0$
\begin{equation}
  \label{eq: bound u_o}
  u_o-u_{o,\infty}^\pm\in E^\nu_{4}(\RR\setminus[-1,1],\RR)
  \text{ with }
  u_{o,\infty}^\pm(\s):=\text{sgn}(\s)\cos(k_0\s).
\end{equation}
As in Step 1, there is no work to be done; indeed, the existence of such a function $u_o$ is, as in~\cite{Buffoni2017a}, obvious:
choose any odd smooth $u_0$ that vanishes in a neighbourhood of $0$ and that is equal to $\sgn(x)\cos(\k_0\s)$ outside another,
larger, neighbourhood.

\emph{Step 3: ``Almost solution'' family for $\epsilon>0$.}  As indicated before, we shall build from $w_{0,\beta}$ a particular
family of functions $[-1,1]\ni \beta\rightarrow w_\beta\in W^{2,\infty}(\RR) \cap C^2(\RR\backslash\{0\})$ which are odd in $\s$
and such that $[-1,1]\ni \beta\rightarrow w_\beta\in W^{2,\infty}(\RR)$ is $C^1$ in $\beta$, for $\epsilon>0$ small
enough. Additionally, $w_{\beta}$ will satisfy $\text{sgn}(w_\beta(\s))=\text{sgn}(\s)$ on $\RR$, $w'_\beta(0)>0$, $w_\beta(\s)$
will tend to a positive periodic solution $w_{\beta,\infty}^+$ to the equation $Lv-\alpha \psi'(v)=0$ as $\s\rightarrow \infty$
and $w_\beta(\s)$ will tend to a negative periodic solution $w_{\beta,\infty}^-$ as $\s\rightarrow -\infty$. For simplicity, we do
not note explicitly the dependence of $w_\beta$ on $\epsilon$.

For small $\epsilon>0$, to obtain $w_\beta$ from $w_{0,\beta}$, we use, with some modifications, the centre manifold theory
developed by Iooss and Kirchg\"assner~\cite{Iooss2000a}, in our case applied near the constant solutions $\pm 1$. Moreover
$w_\beta=u_p$ in a neighbourhood of $0$ independent of $\beta\in[-1,1]$ and small $\epsilon>0$. This centre manifold argument is
presented in Section~\ref{sec:Construction-asymp}. We do not perform a normal form reduction as Iooss and Kirchg\"assner, but
instead parametrise the centre manifold and obtain a ``homotopy'' that allows us to construct approximate solutions $w_\beta$.

\emph{Step 4. Existence proof.} We shall then study the existence of $\beta\in[-1,1]$ and a ``corrector function'' $r$ in an
appropriate space of bounded functions such that $w_\beta-r$ is a solution to the equation
\begin{equation}
  \label{eq:r}
  c^2\Big(w_\beta-r\Big)''
  -\Delta_D \Big(w_\beta -r\Big)
  \\+\alpha\Big(w_\beta-r\Big)
  -\alpha\psi'\Big(w_\beta-r\Big) 
  =0.
\end{equation}

The outline of the remaining arguments is as follows. In Section~~\ref{sec:Construction-asymp}, we will prove the existence of the
family $w_\beta$ of ``approximate'' solutions used in Step 3; the argument relies on centre manifold theory. Properties of this
family of functions are established in Section~\ref{sec:Prop-family-wbeta}. Section~\ref{sec:Exist-heter-conn} contains the fixed
point argument used in Step 4 and thus finishes the proof.

\section{Construction of asymptotic wave trains}
\label{sec:Construction-asymp}

Since the first two steps of the proof strategy outlined in the previous section entirely rely on existing results, we now focus
on Step 3.  Specifically, we construct a family $w_\beta$ of wave trains which have asymptotically the correct behaviour, in the
sense that they solve~\eqref{eq:main} as $\s\to\pm\infty$. This is the key step in the argument, as the anharmonicity of the wells
of the on-site potential is now crucial. We use centre manifold theory. Note that $w_\beta$ are only approximate solutions
to~\eqref{eq:main}; for $\beta\in\{-1,1\}$ they will typically differ significantly from solutions near the dislocation site
$\s=0$, and be only asymptotically correct for large values of $|\s|$ (see the third part of Proposition~\ref{prop: behaviour}).

In Section~\ref{sec:Exist-heter-conn}, we prove the existence of a corrector $r$ such that $w_\beta - r$ solves~\eqref{eq:main}
or, equivalently,~\eqref{eq:r}. We remark that the symmetry of the problem is important here. In essence, $w_\beta$ glues together
two wave trains, one as $\s\to-\infty$ oscillating in the well of $\psi$ centred at $-1$, and one as $\s\to\infty$ oscillating in
the well centred at $1$.

Let us now state the main result of this section, the proof of which is postponed to the end
(Subsection~\ref{sec:Proof-Theor-H1}). Throughout this section, the standing assumptions are those made in
Theorem~\ref{theo:main}.

The main aim is to prove the existence of the ``approximate'' solutions $w_\beta$ proposed in Step 3 in the previous section.
This is a nontrivial problem, as we require these functions to be asymptotic to periodic solutions (Item~\ref{it:H1-asym} in the
following theorem). The results establishes the existence of a function $H_1$ which maps $w_{0,\beta}$ and its derivative to
$w_\beta$, in a pointwise manner. We recall the definition~\eqref{eq: def of wob} of $w_{0,\beta}$, and that $u_p$ is the solution
to~\eqref{eq:kz}.
\begin{theorem}
  \label{thm:def of H1}
  For all $\epsilon>0$ small enough, there exists $H_1\in C^{4}(\RR^2)$ and a period map
  $\widetilde \sP\in C^{4} ([-1,1],(0,\infty))$, both depending on $\epsilon$, such that
  \begin{enumerate}
  \item \label{it:def-H1:1} $H_1(\z,\xx)-\z\rightarrow 0$ in $W^{4,\infty}(\RR^2)$ as $\epsilon\rightarrow 0$.
  \item \label{it:def-H1:2} $H_1(\z,\xx)$ is odd in $\z$.
  \item \label{it:def-H1:3} $H_1(\z,\xx)=\z$ on $(-\epsilon_0/2,\epsilon_0/2)\times \RR$ for some $\epsilon_0>0$
    independent of $\epsilon$.
  \item \label{it:H1-asym} With $\ts:=\dfrac{2\pi\s}{\widetilde \sP(\beta)k_0}$, the function
    $w_\beta(\s):=H_1\Big(w_{0,\beta}(\ts), w_{0,\beta}'(\ts)\Big)$ is asymptotic to a positive periodic solution
    to~\eqref{eq:main} of period $\widetilde \sP(\beta)$ as $\s\rightarrow +\infty$.
  \item \label{it:H1-sign} Furthermore, $\text{sgn}(w_\beta(\s))=\text{sgn}(\s)$ for all $\s\in\RR$,
    $w_\beta(\s)=u_p(\ts)$ in a neighbourhood of $0$ independent of $\beta\in[-1,1]$ and small $\epsilon>0$, and
    $w'_\beta(0)=\dfrac{2\pi}{\widetilde \sP(\beta)k_0}u'_p(0)>0$.
  \item \label{it:H1-P} $\widetilde \sP(\beta)\rightarrow 2\pi/k_0$ and
    $\frac d{d\beta}\widetilde \sP(\beta)\rightarrow 0$ uniformly in $\beta\in[-1,1]$ as $\epsilon\rightarrow 0$.
  \end{enumerate}
\end{theorem}

\emph{Remark.} Additional assumptions on higher order derivatives of the map $u\rightarrow \psi_\epsilon(u)$ for
$|u|\geq \epsilon$ would allow higher-order convergence in claim~\ref{it:def-H1:1}.  Note, however, that the third part implies
that $H_1$ is not only $C^4$ but even smooth on $(-\epsilon_0/2,\epsilon_0/2)\times\RR$.

The proof of Theorem~\ref{thm:def of H1} is given in Subsection~\ref{sec:Proof-Theor-H1}.

\subsection{Centre manifold analysis}
\label{sec:Centre-manif-analysi}

As preparation, we perform a centre manifold analysis, following closely~\cite{Iooss2000a} (as excellent other source on the
centre manifold approach for lattice systems, we refer the reader to~\cite{James2003a}). Small modifications are required, since
the analysis in~\cite{Iooss2000a} is local, while we need a global result, as dislocation waves have large oscillations. 

To relate easily to the notation of~\cite{Iooss2000a}, we write the on-site potential $g$ of~\eqref{eq:g} as $g(u) = \alpha V(u)$,
hence $V(u) = \frac 1 2 u^2 - \psi(u)$. As explained after~\eqref{eq:g}, for the special (limit) case $\psi'(u) = \sgn(u)$, the
on-site potential becomes
\begin{equation*}
  g(u) = \alpha V(u)=g_0(u)=\frac \alpha 2 \min\left((u+1)^2,(u-1)^2\right) =\frac \alpha 2 (u^2+1)-\alpha |u|,
\end{equation*} 
a primitive of $\alpha u - \alpha \sgn (u)$. Obviously, by adding a constant, the equilibrium $u=1$ can be shifted to $u=0$ (and
so can the equilibrium $u=-1$).  In a neighbourhood of $0$, the (shifted) potential becomes simply
$\alpha V(u)=\frac \alpha 2 u^2$ and thus $u-V'(u)=0$. An analogous remark also holds for $\psi_\epsilon$: the corresponding
on-site potential can be chosen to be
\begin{equation*}
  \alpha V(u)
  =\frac \alpha 2 (u^2+1)-\alpha\psi_\epsilon(u)
  =\frac \alpha 2 \min\left((u+1)^2,(u-1)^2\right)
  +\alpha\left(|u|-\psi_\epsilon(u)\right).
\end{equation*}
By shifting the equilibrium $u=1$ to $u=0$, the (shifted) potential becomes
\begin{equation*}
  \alpha V(u)
  =\frac \alpha 2 u^2
  +\alpha\left(|1+u|-\psi_\epsilon(1+u)\right)
\end{equation*}
in a neighbourhood of $0$. The corresponding expression $u-V'(u)$ becomes
\begin{equation*}
  \psi'_\epsilon(1+u)-\sgn(1+u); 
\end{equation*}
it is under control in a neighbourhood of $0$ (in the shifted potential) when $\epsilon\rightarrow 0$,
by~\eqref{eq:psiprimeoutside}. In what follows, we have in mind such a situation in which $u-V'(u)$ is under control in a
neighbourhood of $0$. In particular, as in~\cite{Iooss2000a}, we can assume that (after a shift) there is an equilibrium at $0$.

In~\cite{Iooss2000a}, the governing equation is written as
\begin{equation}
  \label{eq:IK}
  \ddot \z(t)+\tau^2 V'(\z(t))=\gamma \tau^2[\z(t-1)-2\z(t)+\z(t+1)],
\end{equation}
where $\z(t)$ stands for $u(\s)$, and $\tau,\gamma>0$ are given by
\begin{equation}
  \label{eq:taugamma}
\tau^2:=\alpha/c^2,
\qquad
\gamma:=1/\alpha.
\end{equation} 
We follow the notation of~\cite{Iooss2000a} for a while, partially since it is convenient to have a potential $V$ with a minimum
at the origin. Later we will translate the results to our setting and notation, and thus transplant the results to the two
different equilibria $\pm 1$. The potential $V$ is assumed to be of class $C^{\m+1}$ for some $\m\geq 1$ (later, we will only
require $\m =4$). Equation~\eqref{eq:IK} is then rewritten as
\begin{equation}
  \label{eq: transformed equation}
  \partial_tU=L_{\gamma,\tau}U+M_\tau(U)
\end{equation}
with 
\begin{equation*}
  U(t)(\v)=(\z(t),\xx(t),\Z(t,\v))^T,~~\Z(t,0)=\z(t),~~\v\in[-1,1],
\end{equation*}
\begin{equation*}L_{\gamma,\tau}=
  \begin{pmatrix}
    0&1&0\\
    -\tau^2(1+2\gamma)&0&\gamma\tau^2(\delta^1+\delta^{-1})
    \\0&0&\partial_\v
  \end{pmatrix},
\end{equation*}
where $\delta^{\pm 1}$ stands for the evaluation at $s=\pm 1$, and
\begin{equation*}
  M_\tau(U)=\tau^2(0,\z-V'(\z),0)^T~~\text{ with }~~
  U=(\z,\xx,\Z(\cdot))^T.
\end{equation*}
Note that $u-V'(u)$ vanishes at $u=0$.

As in~\cite{Iooss2000a}, let $\HH$ and $\DD$ be Banach spaces for $U=(\z,\xx,\Z(\cdot))^T$,
\begin{align*}
  \HH &:=\RR^2\times C[-1,1] \\
\intertext{and} 
  \DD &:=\{U =(\z,\xx,\Z(\cdot))^T \in \RR^2\times C^1[-1,1]: \Z(0)=\z\},
\end{align*}
equipped with the usual maximum norms, respectively.  The map $ M_\tau$ is $\Cm(\DD,\DD)$.

Let the reflection $S$ in $\HH$ be defined by
\begin{equation*}
  S(\z,\xx,\Z)^T:=(\z,-\xx,\Z\circ \rho)^T,~\text{with}~\rho(\v):=-\v,
\end{equation*}
and note that $L_{\gamma,\tau}$ and $M_\tau$ anticommute with $S$ (``reversibility'').

We denote by $\Delta_0$ the set of pairs $(\gamma,\tau)$ such that the part of the spectrum of $L_{\gamma,\tau}$ that lies in
$i\RR$ contains only one pair of simple eigenvalues (they have to sum up to $0$, thanks to reversibility).

In our setting, $(\gamma,\tau)\in \Delta_0$, since for $c \leq 1$ and close to $1$ the dispersion function $D$ has exactly two
roots, as shown in Section~\ref{sec:Setting-main-result}, and the roots are not degenerate. We denote by $P_1$ the projection onto
the two-dimensional eigenspace related to the two eigenvalues in $i\RR$ and set $Q_h:=I-P_1$ (more precisely, the kernel of
$P_1$ is the linear space defined by equations~\eqref{eq: G_0} and~\eqref{eq: G_1} below for fixed $t$).

Iooss and Kirchg\"assner refer to Theorem~3 in~\cite{Vanderbauwhede1992a} to prove their theorem about the existence of a
\emph{local} centre manifold. That is, under different conditions on $\mathrm{id}-V'$ from those in Theorem~\ref{theo:vdw} below,
there is a neighbourhood $\Omega$ of $0$ in $\DD$ such that the result holds for $\tilde U_c\colon\RR\rightarrow \Omega_c$ rather
than $\DD_c$ in claim~\ref{it-vdw1} and $\tilde U \colon \RR\rightarrow\Omega$ in claim~\ref{it-vdw2}. Instead, by referring to
Theorem~2 in~\cite{Vanderbauwhede1992a}, one gets in the same way the following theorem (\emph{global} centre manifold)\footnote{
  We apply Theorem~2 in~\cite{Vanderbauwhede1992a} when $g\in \Cm_b(X;Y)$ (with the notations $g,X,Y$ as
  in~\cite{Vanderbauwhede1992a}), which makes the proof in~\cite{Vanderbauwhede1992a} shorter. Also, still with the notations
  of~\cite{Vanderbauwhede1992a}, $Y=X$ in our setting.  The assumptions of Theorem~2 in~\cite{Vanderbauwhede1992a} are checked for
  completeness in Appendix~\ref{sec:Appl-cent-manif}, following the ideas in~\cite{Iooss2000a}.  }.

\begin{theorem}
  \label{theo:vdw}
  Given $(\gamma,\tau)\in \Delta_0$, assume that $\mathrm{id}-V' \in \Cm_b(\RR)$ (the function and its $m$ first derivatives are
  bounded), $\mathrm{id}-V'$ is Lipschitz continuous and that the Lipschitz constant is small enough (in a way that can depend in
  particular on $\m$).  Then there exists a mapping $h\in \Cm_b(\DD_c,\DD_h)$, where $\DD_c:=P_1\DD$ and $\DD_h:=Q_h\DD$, and a
  constant $p_0>0$ such that the following is true.
  \begin{enumerate}
  \item \label{it-vdw1} If $\widetilde U_{c}\colon \RR\rightarrow \DD_{c}$ is a solution of~\eqref{eq: reduced
      transformed equation},
    \begin{equation}
      \label{eq: reduced transformed equation}
      \partial_t U_c=L_{\gamma,\tau}U_c+P_1M_\tau[U_c+h(U_c)],
    \end{equation}
    then $\widetilde U=\widetilde U_c+h(\widetilde U_c)$ solves~\eqref{eq: transformed equation}.
  \item \label{it-vdw2} If $\widetilde U$ solves \eqref{eq: transformed equation} for all $t\in \RR$ and
    $||e^{-\eta|\cdot|}\widetilde U||_{L^\infty(\RR)}<\infty$ for some $\eta$ in $(0,p_0)$, then
    \begin{equation*}
      \widetilde U_h(t)=h(\widetilde U_c(t)),~~t\in\RR,
    \end{equation*}
    holds with $\widetilde U_c=P_1 \widetilde U$ and $\widetilde U_h=Q_h \widetilde U$, and $\widetilde U_c(t)$
    solves~\eqref{eq: reduced transformed equation}.
\end{enumerate}
\end{theorem}
This global aspect is relevant to our setting, since the centre manifold theory will be applied in large neighbourhoods of the
equilibria $1$ and $-1$ (but small enough to exclude a small neighbourhood of the origin, where the convexity of the on-site
potential fails). In Theorem~\ref{theo:vdw}, it is not necessary to suppose that the origin is an equilibrium. However, as we
assumed for simplicity that $\psi'_\epsilon-u$ vanishes at $u=1$ and $u=-1$, it turns out that the additional property $V'(0)=0$
holds. Moreover, if $M_\tau=0$, then $h=0$ and $U=0$ is obviously the unique equilibrium.

Inspecting the proof of Theorem~2 in~\cite{Vanderbauwhede1992a}, one sees that the norm of $h\in \Cm_b(\DD_c,\DD_h)$ tends to $0$
when the norm of $M_\tau \in C_b^m(\DD,\DD )$ tends to $0$\footnote{ As explained on bottom of page 131 in~\cite{Iooss2000a}, the
  derivatives of $\Psi$ (in the notations of~\cite{Iooss2000a}) can be calculated by formal differentiations of the identity~(11)
  in~\cite{Iooss2000a}.  This gives estimates of the norms of the derivatives of $\Psi-I$ (acting between appropriate Banach
  spaces depending on the order of differentiation) in terms of the norms of the derivatives of $g$ (still in the notations
  of~\cite{Iooss2000a}).}.  Moreover, $h$ commutes with the reversibility operator $S$, so that the reduced equation~\eqref{eq:
  reduced transformed equation} is reversible (as explained at the end of Section~2.2 in~\cite{Vanderbauwhede1992a}). In
particular, it is easily checked that $P_1S=SP_1$ on $\HH$ (see Lemma~2 in~\cite{Iooss2000a} where $P_1$ is explicitly given).
Hence, for $U_c\in \DD_c$, we have $S(U_c+h(U_c))=U_c+h(U_c)$ exactly when $S U_c=U_c$.

The two-dimensional space $\DD_c$ is given by
\begin{equation*}
  \DD_{c}=\{(\delta_1,\delta_2 \k_0,\delta_1\cos(\k_0\cdot)+\delta_2\sin(\k_0\cdot))
  :\delta_1,\delta_2\in \RR\},
\end{equation*}
where $\k_0>0$ is such that $\pm i \k_0$ is in the spectrum of $L_{\gamma,\tau}$ (and there are no other purely imaginary values
in the spectrum).  A simple computation shows that $i\k$ is an eigenvalue with $\k\in \RR\backslash\{0\}$ if and only if
$-\tau^2(1+2\gamma)+\gamma\tau^22\cos(\k)=-\k^2$. When $M_\tau=0$ (and thus $h=0$), the two-dimensional linear space $\DD_{c}$ is
filled by $0$ and the orbits of a smooth one-parameter family of reversible periodic solution
\begin{equation}
  \label{eq:ua}
  t\rightarrow U_a(t) :=(a\cos(\k_0 t),-ak_0\sin(\k_0t),a\cos(\k_0(t+\cdot))\,),
\end{equation}
with $a>0$ being the amplitude.  So, in essence the centre space is parametrised by the amplitude $a$. Each of these periodic
solutions meets the \emph{reversibility line}
\begin{equation*}
  \{\z_c\in \DD_{c}: S\z_c=\z_c\}
  =\{(\delta_1,0,\delta_1\cos(\k_0\cdot)):\delta_1\in \RR\}
\end{equation*} 
twice in one of its periods (at $t=0$ and $t=\pi/k_0$ in the period $[0,2\pi/k_0)$). The intersection with the reversibility line
is transverse: at $t=0$ (say)
\begin{multline*}
  L_{\gamma,\tau}(a,0,a\cos(\k_0\cdot))
  =(0,-a\k_0^2,-a\k_0\sin(\k_0\cdot))
  \\ \neq (0,a\k_0^2,a\k_0\sin(\k_0\cdot))=
  SL_{\gamma,\tau}(a,0,a\cos(\k_0\cdot))
\end{multline*}
for all $a>0$.  Let us restrict the amplitude parameter $a$ to any fixed compact interval $[a_1,a_2]\subset(0,\infty)$ with
$a_1<a_2$.

Iooss and Kirchg\"assner proceed by carrying out a normal form analysis. We proceed differently and give a parametrisation of the
centre manifold, with the amplitude $a$ and the time $t$ being the parameters.

\begin{proposition}
  \label{prop: G}
  For $(a,t)\in [a_1,a_2]\times \RR$, define $G(a,t)\in\DD$ as the value at time $t$ of the solution on the centre manifold that
  starts at time $0$ at $U_a(0)+h(U_a(0))$, with $h$ given by Theorem~\ref{theo:vdw} and $U_a$ as in~\eqref{eq:ua}.  Then $G$ is
  of class $\Cm$ and, when $M_\tau$ tends to $0$ in $\Cm_b(\DD,\DD)$, the map $(a,t)\rightarrow G(a,t)-U_a(t)$ tends to $0$ in
  $\Cm_b([a_1,a_2]\times [-t_1,t_1])$ for all $t_1>0$.

  Moreover, for all $a\in[a_1,a_2]$, $t\rightarrow G(a,t)$ is a reversible periodic solution to~\eqref{eq: transformed equation}.
  The corresponding period $\sP_a$ is a $\Cm$-function of $a\in[a_1,a_2]$ and, when $M_\tau$ tends to $0$ in $\Cm_b(\DD,\DD)$, the
  map $a\rightarrow \sP_a$ tends to the constant map $2\pi/k_0$ in $\Cm_b [a_1,a_2] =\Cm [a_1,a_2] $.
\end{proposition}

\proof The fact that $G$ is $\Cm$ with respect to $(a,t)$ relies on standard results on dependence of solutions with respect to
parameters in finite dimensional dynamical systems (see, e.g., the remarks at the end of Chapter~I, Section~7
in~\cite{Coddington1955a}). The dynamics on the centre manifold is indeed finite dimensional, see~\eqref{eq: reduced transformed
  equation}.

To show that $G(a,t)-U_a(t)$ tends to zero as $M_\tau$ tends to zero, we argue by contradiction. Suppose that
$||M_{\tau,n}||_{\Cm_b(\DD,\DD)} \leq 2^{-n^2}$ and $||h_n||_{\Cm_b(\DD_c,\DD_h)} \leq 2^{-n^2}$ for all $n\geq 0$, with $h_n$
given by Theorem~\ref{theo:vdw} applied to $M_{\tau,n}$, while the corresponding $G_n(a,t)-U_a(t)$ does not tend to $0$ in the
sense above.  Introduce a $\Cm$-interpolation $\widetilde M_\tau(\cdot;\mu)$ of the sequence $\{M_{\tau,n}\}_{n\geq 0}$ such that
$0\leq \mu\leq 1$, $\widetilde M_\tau(\cdot;0)=0$ and $\widetilde M_\tau(\cdot;2^{-n})=M_{\tau,n}$ for all $n\geq 0$. In the same
way, introduce a $\Cm$-interpolation $\widetilde h(\cdot;\mu)$ of the sequence $\{h_n\}_{n\geq 0}$ such that $0\leq \mu\leq 1$,
$\widetilde h(\cdot;0)=0$ and $\widetilde h(\cdot;2^{-n})=h_n$ for all $n\geq 0$. For $\mu$ different from $0$ and not of the form
$2^{-n}$, we shall only consider the reduced equation \eqref{eq: reduced transformed equation} without needing the full equation
\eqref{eq: transformed equation} and therefore we do not need to ensure that $\widetilde h(\cdot,\mu)$ corresponds to
$\widetilde M_{\tau}(\cdot,\mu)$.

For $(a,t)\in [a_1,a_2]\times \RR$ and $\mu\in[0,1]$, define $\widetilde G(a,t;\mu)\in\DD$ as the value at time $t$ of the
solution of \eqref{eq: reduced transformed equation} that starts at time $0$ at $U_a(0)+\widetilde h(U_a(0);\mu)$.  Then
$\widetilde G$ is a $\Cm$-interpolation of the sequence $\{G_n\}_{n\geq 0}$ such that $\widetilde G(a,t;0)=U_a(t)$ and
$\widetilde G(\cdot,\cdot;2^{-n})=G_n$ for all $n\geq 0$.  As $\widetilde G$ and all its derivatives up to order $\m$ are
continuous at any $(a,t,\mu)$ with $\mu=0$, $\widetilde G(\cdots;2^{-n})$ converges to $\widetilde G(\cdots;0)$ in
$\Cm_b([a_1,a_2]\times[-t_1,t_1])$ for all $t_1>0$. This is a contradiction, as we have supposed \emph{ad absurdum} that
$G_n(a,t)-U_a(t)$ does not tend to $0$.

The function $t\rightarrow G(a,t)$ is reversible and periodic of half-period $\sP_a/2>0$ exactly when $SG(a,t)-G(a,t)$ vanishes at
$t=\sP_a/2$ without vanishing on $(0,\sP_a/2)$.  In this case, $\sP_a/2$ and $\sP_a$ satisfy the equation
$E_2P_1G(a,\sP_a/2)=E_2P_1G(a,\sP_a)=0$, where $E_2$ is the projection on the second real component of a vector in $\DD$ (and we
recall that $P_1$ is the projection on $\DD_c$). If $M_\tau=0$, then $\sP_a=2\pi/k_0$ for $a\neq 0$
\begin{equation*}
  \frac{d}{dt}E_2P_1G(a,t)|_{t=\sP_a/2}=a\k_0^2\neq 0
  ~~\text{ and }~~\frac{d}{dt}E_2P_1G(a,t)|_{t=\sP_a}=-a\k_0^2\neq 0. 
\end{equation*}
Hence, if $M_\tau\in \Cm_b(\DD,\DD)$ is small enough, $t\rightarrow G(a,t)$
is a reversible periodic orbit,
\begin{equation*}
  \frac{d}{dt}E_2P_1G(a,t)|_{t=\sP_a/2}\neq 0,~~
  \frac{d}{dt}E_2P_1G(a,t)|_{t=\sP_a}\neq 0 
\end{equation*}
for all $a\in[a_1,a_2]$ and, by the implicit function theorem, $\sP_a$ is a $\Cm$-function of $a$. When $M_\tau=0$, $\sP_a$ is
equal to $2\pi/k_0$.  Hence, still by the implicit function theorem, the map $a\rightarrow \sP_a$ tends to the constant map
$2\pi/k_0$ in $\Cm_b[a_1,a_2]$ when $M_\tau$ tends to $0$ in $\Cm_b(\DD,\DD)$.  \qed

We now make the main step in establishing the existence of the function $H_1$ in Theorem~\ref{thm:def of H1}. We remark that first
component $H_1$ of the function $H$ discussed in the following proposition will (with minimal modifications summarised in
Proposition~\ref{prop: function H around 1}) be a restriction of the function $H_1$ of Theorem~\ref{thm:def of H1}.

\begin{proposition}
  \label{prop:function H}
  Let $0<a_1<a_2$. There exists a $\Cm_b$ map
  \begin{equation*}
    (\z,\xx)\rightarrow H(\z,\xx)= (H_1(\z,\xx),H_2(\z,\xx),H_3(\z,\xx))     \in \DD
  \end{equation*} 
  defined for pairs $(\z,\xx)\in \RR^2$ which satisfy $\displaystyle a_1\leq \sqrt{\z^2+k_0^{-2}\xx^2}\leq a_2$, with the
  following properties.  If $a\in[a_1,a_2]$, the set
  \begin{equation*}
    \left\{\Big(a\cos(k_0t),-ak_0\sin(k_0t)\Big):t\in \RR\right\}
  \end{equation*}
  belongs to the domain of $H$.  The map
  \begin{equation*}t\mapsto H\Big(a\cos(2\pi t/\sP_a), -ak_0\sin(2\pi t/\sP_a)\Big)
  \end{equation*}
  is a $\sP_a$-periodic and reversible solution to~\eqref{eq: transformed equation} (or, equivalently, its first component
  solves~\eqref{eq:IK}) on the centre manifold. When $||M_\tau||_{\Cm_b(\DD,\DD)}$ tends to $0$, the map $(H_1,H_2)$ tends to the
  identity map in the $\Cm$-norm (on the domain of $H$), $\sP_a\rightarrow 2\pi/k_0$ and $\frac d{da}\sP_a\rightarrow 0$ uniformly
  in $a\in[a_1,a_2]$.
\end{proposition}

\proof Let $0<a_1<a_2$. By Proposition~\ref{prop: G}, for $a\in[a_1,a_2]$, the function $t\rightarrow G(a,t) \in \DD$ is a
reversible periodic solution to~\eqref{eq: transformed equation} with period $\sP_a>0$; it can be parametrised by $a$ and
\begin{equation*}
  \widetilde t:=2\pi t/(\sP_ak_0)\in \RR.
\end{equation*} 
In the variable $\widetilde t$, the period is independent of $a$ and equal to $2\pi/k_0$.  Hence, we obtain a parametrisation of a
compact piece of the centre manifold
\begin{equation*}
  (a,\widetilde t)\rightarrow  \widetilde H(a,\widetilde t)=
  (\widetilde H_1(a,\widetilde t),\widetilde H_2(a,\widetilde t),
  \widetilde H_3(a,\widetilde t))
  := G(a,t)\in \DD
\end{equation*} 
for $a_1\leq a\leq a_2$ and $\widetilde t\in \RR$, which is $2\pi/k_0$-periodic and reversible in $\widetilde t$, i.e.,
$\widetilde H(a,-\widetilde t)= S\widetilde H(a,\widetilde t)$.  This piece of centre manifold is invariant and $a\times \RR$ is
sent to a reversible periodic solution, up to a linear reparametrisation.  When $||M_\tau||_{\Cm_b(\DD,\DD)}$ is small, the map is
near the map
\begin{equation*}
  (a,\widetilde t)\rightarrow U_a(\widetilde t)=(a\cos(k_0 \widetilde t),
  -ak_0\sin(k_0 \widetilde t),a\cos(k_0 \widetilde t+\cdot))
\end{equation*}
and actually equal to this map when $M_\tau=0$.

By Proposition~\ref{prop: G}, the map $a\rightarrow \sP_a$ tends to the constant map $2\pi/k_0$ in $\Cm_b[a_1,a_2]$ when $M_\tau$
tends to $0$ in $\Cm_b(\DD,\DD)$.  Moreover, $\frac d{da}\sP_a=0$ when $M_\tau=0$ (because the period is constant, equal to
$2\pi/k_0$) and therefore $\frac d{da}\sP_a\rightarrow 0$ uniformly in $a\in[a_1,a_2]$ as
$||M_\tau||_{\Cm_b(\DD,\DD)}\rightarrow 0$.

Given $(\z,\xx,\Z)\in \DD_{c}$ in the range of $U_a$, $a$ can be recovered by the formula
\begin{equation}
  \label{eq:a}
  a=a(\z,\xx)=\sqrt{\z^2+\k_0^{-2}\xx^2}
\end{equation} 
and, modulo $2\pi/k_0$, $\widetilde t=\widetilde t (\z,\xx)$ can be recovered from
\begin{equation*}
  (\cos (\k_0\widetilde t),\sin(\k_0\widetilde t))=(a^{-1}\z,-a^{-1}\k_0^{-1}\xx).
\end{equation*}
This gives the desired map $(\z,\xx)\rightarrow H(\z,\xx):=\widetilde H(a(\z,\xx),\widetilde t(\z,\xx))$.  \qed

We now return to our initial notation. Let us focus on the well around $1$. To apply the centre manifold theorem with order of
differentiability $\m=4$, we redefine $\psi=\psi_\epsilon$ on $(-\infty,\epsilon)$ so that $|\psi'_\epsilon(u)-1|< C\, \epsilon$
and~\eqref{eq: cdn on psi} holds on $\RR$ for all small $\epsilon>0$ (changing the value of $C$ if necessary).  We then obtain the
following proposition as reformulation of Proposition~\ref{prop:function H}. Note that $a\cos(\k_0t)$ is replaced by
$1+a\cos(\k_0t)$ to take account of the fact that we are now concerned with the well of $\psi$ centred at $1$.  Moreover, we
revert to writing $\s$ instead of $t$, and write the wave equation~\eqref{eq:IK} again as in~\eqref{eq:main},
\begin{equation*}
  c^2u''-\Delta_Du+\alpha u-\alpha \psi'(u)=0.
\end{equation*}

\begin{proposition}
  \label{prop: function H around 1}
  Let $0<a_1<a_2<1$. For all $\epsilon>0$   small enough, there exists a $\Cm_b$ map
  \begin{equation*}
    (\z,\xx)\rightarrow H(\z,\xx)=
    (H_1(\z,\xx),H_2(\z,\xx),H_3(\z,\xx))
    \in \DD
  \end{equation*} 
  defined for pairs $(\z,\xx)\in \RR^2$ satisfying $\displaystyle a_1\leq \sqrt{(\z-1)^2+\k_0^{-2}\xx^2}\leq a_2$, with the
  following properties.  If $a\in[a_1,a_2]$, the set
  \begin{equation*}
    \left\{\Big(1+a\cos(\k_0\s),-a\k_0\sin(\k_0\s)\Big):\s\in \RR\right\}
  \end{equation*}
  belongs to the domain of $H$.  The map
  \begin{equation*}
    \s\rightarrow H_1\Big(1+a\cos(2\pi \s/\sP_a), -a\k_0\sin(2\pi \s/\sP_a)\Big)
  \end{equation*} 
  is a $\sP_a$-periodic and reversible solution to~\eqref{eq:main}. When $\epsilon \rightarrow 0$, the map $(H_1,H_2)$ tends to
  the identity map in the $\Cm$-norm (on the domain of $H$), $\sP_a\rightarrow 2\pi/k_0$ and $\frac d{da}\sP_a\rightarrow 0$
  uniformly in $a\in[a_1,a_2]$.
\end{proposition}

This proposition in particular establishes the existence of $H_1$, the first component of $H$. We will prove in the following
subsection that a suitable extension of this function has the properties of the function $H_1$ claimed in Theorem~\ref{thm:def of
  H1}.

\subsection{Proof of Theorem~\protect{\ref{thm:def of H1}}}
\label{sec:Proof-Theor-H1}

The function $H_1$ of Proposition~\ref{prop: function H around 1} establishes the existence of wave trains oscillating in the well
centred at $1$. We now extend this function by symmetry to a smooth function that gives anharmonic wave trains oscillating in the
wells at $\pm 1$ as $\s \to \pm \infty$.

The map $(\z,\xx)\rightarrow H_1(\z,\xx)$ of the last proposition sends the function
\begin{equation*}
  \RR\ni \s\rightarrow \Big( 1+a\cos(2\pi \s/\sP_a)\,,\,
  -a\k_0\sin(2\pi \s/\sP_a)\Big)
\end{equation*} 
to a periodic solution of equation~\eqref{eq:main} ($a_1\leq a\leq a_2$).  Analogously, by shifting $x\rightarrow x+\sP_a/2$, the
function
\begin{equation}
  \label{eq:periodfct}
  \RR\ni \s\rightarrow \Big( 1-a\cos(2\pi \s/\sP_a)\,,\,
  a\k_0\sin(2\pi \s/\sP_a)\Big)
\end{equation} 
is sent to a periodic solution, too. As $\epsilon$ is near $0$, $H_1(\z,\xx)$ is near $\z$ by Proposition~\ref{prop: function H
  around 1} and, when $\epsilon=0$, $H_1(\z,\xx)=\z$.  Given $0<a_1<a_2<1$, let $\epsilon_0>0$ be such that $u>\epsilon_0$ and
$H_1(u,v)>\epsilon_0$ for all $(u,v) \in \RR^2$ with $a_1\leq \sqrt{ (u-1)^2 + k_0^{-2}v^2 } \leq a_2$. Since $\psi\in C^5$, we
can assume that $H_1$ is $C^4$; moreover $H_1$ is well-defined on a compact convex subset of $(\epsilon_0,\infty)\times \RR$ with
non-empty interior. We can then extend $H_1-\z$ in a $C^{4}$ way on $\RR^2$, such that $H_1-\z$ is small in $W^{4,\infty}(\RR^2)$;
see~\cite[Paragraph VI 2.3]{Stein1970a}.  The extension can be chosen such that $H_1(\z,\xx)$ is odd in $\z$ and that
$H_1(\z,\xx)=\z$ on $(-\epsilon_0/2,\epsilon_0/2)\times \RR$.  Remembering that $\psi'$ is odd, the analysis around the well $1$
as $\s\rightarrow \infty$ can therefore be transferred to the well $-1$ as $\s\rightarrow -\infty$. This establishes
claims~\ref{it:def-H1:1}--\ref{it:def-H1:3} of Theorem~\ref{thm:def of H1}.

We now turn to the proof of claims~\ref{it:H1-asym} and \ref{it:H1-P} of this theorem.  For $\s\in\RR$, $w_{\beta}(\s)$ has been
defined there as
\begin{equation*}
  w_\beta(\s)=H_1\Big(w_{0,\beta}\Big(2\pi\s /(\sP_a\k_0)\Big),\, 
  w_{0,\beta}'\Big(2\pi\s /(\sP_a\k_0)\Big)\Big),
\end{equation*}
where $w_{0,\beta}$ is given by~\eqref{eq: def of wob} and $\sP_a>0$ is the period corresponding to
\begin{equation}
\label{eq: how to get a}
  a=\lim_{\s\rightarrow +\infty}
  \sqrt{(w_{0,\beta}(\s)-1)^2+\k_0^{-2}w_{0,\beta}' (\s)^2}\, , 
\end{equation}
analogously to~\eqref{eq:a}. The constant $1$ has been subtracted from $w_{0,\beta}(\s)$ since the analysis is carried out around
the constant solution $1$ when $\s\rightarrow +\infty$.

As $a$ is a function of $\beta$, so is $\sP_a$, and we set $\widetilde \sP(\beta)=\sP_a.$ Let us go back to the
definition of $w_{0,\beta}$ in~\eqref{eq: def of wob},
\begin{equation*}
  w_{0,\beta} =u_p+B\beta u_o,
\end{equation*}
where $B>0$, $u_p$ is the particular odd solution of~\eqref{eq:kz} found in~\cite{Kreiner2011a} for $\epsilon=0$, and the odd
function $u_o$ satisfies~\eqref{eq: bound u_o} and vanishes in a neighbourhood of $0$.  In~\cite{Kreiner2011a}, it is shown that
$u_p'(0)>0$, $\text{sgn}(u_p(\s))=\text{sgn}(\s)$ on $\RR$ and $u_p$ converges exponentially to
\begin{equation}
  \label{eq: bound u_p}
  u_{p,\infty}^{\pm}(\s)=
  \pm\left\{1-\frac{c^2\k_0^2-2}{c^2\k_0^2-\k_0}\cos(\k_0\s) \right\}
\end{equation}
as $\s \rightarrow \pm \infty$. Since $Lu_p=\alpha\text{sgn}(x)$ with $u_p\in W^{2,\infty}(\RR)\cap C^2(\RR\backslash\{0\})$, a
simple bootstrapping argument shows that $u_p'''\in C(\RR\backslash\{0\})$ can be continuously extended at $x=0$ and that
$u_p''''\in C(\RR\backslash\{-1,0,1\})$ with left and right limits at $x\in\{-1,0,1\}$.  Moreover, as $\s\rightarrow +\infty$
(respectively $\s\rightarrow -\infty$), the rate of convergence is exponential of the type $e^{-|\nu \s|}$ in the sense that
$u_p-u^{+}_{p,\infty}$ (respectively $u_p-u^{-}_{p,\infty}$) and its four first derivatives are bounded on $\RR$ when multiplied
by $e^{|\nu\s |}$. Moreover the parameter $\nu<0$ can be assumed to satisfy $|\nu|<p_0$, with $p_0$ as in Theorem~\ref{theo:vdw}.
Furthermore, it holds that $0<\inf u_{p,\infty}^+<\sup u_{p,\infty}<1$. We choose $0<a_1<a_2<1$ such that
$a_1<\inf u_{p,\infty}^+<\sup u_{p,\infty}<a_2$. Then $B>0$ in the definition of $w_{0,\beta}$ is chosen small enough so that
$a_1<\inf u_{p,\infty}^+ -B<\sup u_{p,\infty} +B <a_2$. 

Let us ignore for the moment issues of convergence and study the image of $u_{p,\infty}^\pm +B\beta u_{o,\infty}^\pm$ and its
derivative under $H_1$, rather than the image of $w_{0,\beta} =u_p+B\beta u_o$ and its derivative. Note that
\begin{equation*}
  \Big(u_{p,\infty}^\pm(\ts) +B\beta u_{o,\infty}^\pm(\ts),
  {u_{p,\infty}^\pm}'(\ts) +B\beta {u_{o,\infty}^{\pm}}'(\ts)\Big)
\end{equation*}
is of the form~\eqref{eq:periodfct} for small enough constant $B$; thus by Proposition~\ref{prop: function H around 1} for all
$|\beta|\leq 1$
\begin{equation}
 \label{eq:hamster1} 
 \s \rightarrow H_1\left(
   u_{p,\infty}^{\pm}(\ts)\pm   B\beta\cos(\k_0\ts),{u_{p,\infty}^{\pm}}'(\ts)\mp  \k_0 B\beta\sin(\k_0\ts)
 \right),
\end{equation}
with $\ts=2\pi \s/(\sP_a\k_0)$ is a periodic solution to~\eqref{eq:main}, where the period is $\sP_a>0$ and (see~\eqref{eq: how to
  get a})
\begin{equation*}
  a
  =\left|B\beta-\frac{c^2\k_0^2-2}{c^2\k_0^2-\k_0}\right|.
\end{equation*}

We are thus left with studying the convergence of $w_{0,\beta}$, that is, the convergence of $u_p$ to $u_{p,\infty}$.
In~\cite{Kreiner2011a}, $u_p$ is shown to be of the form $u_p=\widetilde u_p-r$, with the following properties.
\begin{enumerate}
\item \label{it:kreiner1} The odd function $\widetilde u_p$ and the Fourier transform $\widehat r$ of $r$ are explicitly given.
\item \label{it:kreiner2} The function $\widetilde u_p\in W^{2,\infty}(\RR)\cap C^\infty(\RR\backslash\{0\})$ converges
  exponentially to $u_{p,\infty}^{\pm}(\s)$ as $\s\rightarrow \pm\infty$, with corresponding exponential convergence of its four
  first derivatives.
\item \label{it:kreiner3} $L\widetilde u_p-\alpha\text{sgn}(\s)$ is continuous at $\s=0$.
\item \label{it:kreiner4} The Fourier transform $k\rightarrow \widehat r(k)$ is smooth and decays with all its derivatives to $0$
  at $\pm\infty$ at least as $|k|^{-5}$. In particular $r\in H^4(\RR)$ and $Lr\in H^2(\RR)$.
\item \label{it:kreiner5} The identity $\text{sgn}(\widetilde u_p(\s))=\text{sgn}(u_p(\s))=\text{sgn}(\s)$ holds on $\RR$.
\end{enumerate}
As $Lr=L\widetilde u_p-\alpha\text{sgn}(\s)$ decays exponentially to $0$, so do $r$ and $r'$ by Proposition~\ref{prop: on the
  exponential decay}. As $c^2r''=\Delta_D r-\alpha r+L\widetilde u_p-\alpha\text{sgn}(\s)$ and the two first derivatives of
$L\widetilde u_p-\alpha\text{sgn}(\s)$ decays exponentially, so do $r''$, $r'''$ and $r^{(4)}$.  As
$c^2u_p''=\Delta_Du_p-\alpha u_p+\alpha\text{sgn}(\s)$, we get (as already observed) that
$u_p'''\in C(\RR\backslash\{0\})\cap L^{\infty}(\RR)$ and that $u_p'''$ can be continuously extended at $x=0$.

This decay, in combination with property~\ref{it:kreiner2} and the fact that~\eqref{eq:hamster1} defines, as just shown, a
periodic solution to~\eqref{eq:main}, establishes claims~\ref{it:H1-asym} and~\ref{it:H1-P} of Theorem~\ref{thm:def of
  H1}. Finally, claim~\ref{it:H1-sign} of the theorem follows immediately from the fact that $\text{sgn}(u_p(\s))=\text{sgn}(\s)$
on $\RR$ with $u_p'(0)>0$, as shown in~\cite{Kreiner2011a}; then $\text{sgn}(w_{0,\beta}(\s))=\text{sgn}(\s)$ on $\RR$ with
$w_{0,\beta}'(0)>0$ by~\eqref{eq: def of wob}; the result follows since by claim~\ref{it:def-H1:2} $H(u,v)$ is odd in $u$ and by
claim~\ref{it:def-H1:3} $H_1(\z,\xx)=\z$.  \qed

\section{Properties of the family $w_\beta$}
\label{sec:Prop-family-wbeta}

In this section, we establish various properties of the family $w_\beta$ which will be used in the fixed point argument in
Section~\ref{sec:Exist-heter-conn} to complete the proof of Theorem~\ref{theo:main}. Throughout this section, $w_\beta$ will be as
defined in Theorem~\ref{thm:def of H1}. Let $\XX : = \{ f\in L^\infty(\RR): e^{ -\nu |x|}f\in L^\infty(\RR) \}$.

\begin{proposition}
  \label{prop: behaviour}
  For all $\nu<0$ close enough to $0$, the following holds.
  \begin{enumerate}
  \item \label{it:behaviour1} 
    With $L$  defined in~\eqref{eq:L}, 
    \begin{equation*}
      Lw_\beta-\alpha     \psi'(w_\beta)
      \in \XX \cap E^\nu_0(\RR\backslash[-2,2], \RR)
    \end{equation*}
    and it is in $C([-1,1], \XX)\cap C^1([-1,1],E^\nu_0(\RR\backslash[-2,2], \RR))$ as a function of $\beta\in [-1,1]$.
  \item \label{it:behaviour2} The map $\beta\rightarrow \int_\RR \Big(Lw_\beta-\alpha \psi'(w_\beta)\Big)\sin(\k_0\s)d\s$ is $C^1$
    and the transversality condition
    \begin{equation}
      \label{eq: K_0}
      \frac{d}{d\beta}\int_\RR\Big(Lw_\beta-\alpha     \psi'(w_\beta)\Big)\sin(\k_0\s)d\s
      \neq 0
    \end{equation}
    holds for all $\beta\in[-1,1]$.
  \item \label{it:behaviour3} In addition, we have
    \begin{equation*}
      \lim_{(\beta,\epsilon)\rightarrow 0}\sup\left\{
        e^{|\nu \s|}\left|(Lw_{\beta})(\s)-\alpha\psi'(w_\beta(\s))\right|:
        \s\in \RR, |w_{\beta}(\s)|\geq \epsilon\right\}=0,
    \end{equation*}
    where $w_{\beta}$ and $\psi$ (but not $\nu$) depend on $\epsilon>0$.
  \end{enumerate}
\end{proposition}

In this proposition, the sign convention $\nu<0$ is chosen to be consistent with the notations
in~\cite{Iooss2000a,Vanderbauwhede1992a}.  Moreover, claim~\ref{it:behaviour1} ensures that the integral in~\eqref{eq: K_0} is
well defined.  In the proof, we shall see that in fact
\begin{equation*}
  \kk_0:=\inf_{\beta\in[-1,1]}
    \left|  \frac{d}{d\beta}\int_\RR(
       (L w_\beta
      -\alpha\psi'(w_\beta))
      \sin(\k_0\s) d\s\right|
\end{equation*}
is uniformly bounded below by a positive constant that does not depend on small $\epsilon>0$.

\proof In this proof, $B>0$ is chosen as small and $\nu<0$ as close to $0$ as required (but in a way that is independent of
$\epsilon$ small).  Let us reconsider the function $w_\beta$ from Theorem~\ref{thm:def of H1} and show that we have exponentially
attained limits.  To this behalf, we write
\begin{equation}
  \label{eq: w_beta}
  w_\beta(\s)=H_1
  \Big(w_{0,\beta}(\ts),w_{0,\beta}'(\ts)\Big)
  \text{\quad with\quad}
  w_{0,\beta}= u_p+B\beta u_o
\end{equation}
for $\s\in \RR$ with $\ts=2\pi\s/(\widetilde \sP(\beta)\k_0)$ and $\widetilde \sP(\beta)$ as in Theorem~\ref{thm:def of H1}.

We write $u_{p,\infty}^\pm$, $u_{o,\infty}^\pm$ and $w_{\beta,\infty}^\pm$ for the corresponding asymptotic periodic functions as
$x\rightarrow \pm\infty$.  Let us choose $C_0>0$ large enough so that $|w_{\beta}|\geq \epsilon$ on
$\RR\backslash(-C_0\epsilon,C_0\epsilon)$ for all $\beta\in[-1,1]$ and all small $\epsilon$. We recall that
$w_{\beta,\infty}^+\geq \epsilon$ and $w_{\beta,\infty}^-\leq -\epsilon$ on $\RR$ if $\epsilon>0$ is small enough.

We obtain by continuity of $H_1$
\begin{equation*}
  w_{\beta,\infty}^\pm(\s)=H_1\Big(\a(\ts),\a'(\ts)\Big)
  \text{\quad with\quad}
  \a= u_{p,\infty}^{\pm}+B\beta u_{o,\infty}^\pm.
\end{equation*}
To prove claim~\ref{it:behaviour1}, we first prove an auxiliary statement.  Let $\b:=w_{0,\beta}-\a$, so that by the fundamental
theorem of calculus
\begin{multline}
  \label{eq: the diff}
  w_\beta(\s)-w_{\beta,\infty}^\pm(\s)
  =\int_0^1
  \b\partial_1H_1\left( \a+\si \b ,  \a'+\si \b'  \right)\\+\b'\partial_2H_1\left( \a+\si \b ,  \a'+\si \b' \right) d\si, 
\end{multline}
where the functions in the arguments of $H_1$ are evaluated at $\ts$.  This expression is exponentially decaying as
$\s\rightarrow \pm\infty$ since $\b$ and $\b'$ decay exponentially, by~\eqref{eq: bound u_p} and~\eqref{eq: bound u_o}. The fact
that left-hand side is evaluated at $\s$ and the right-hand side at $\ts$ is not a problem since we can decrease $|\nu|$. We also
get that $\left(w_\beta(\s)-w_{\beta,\infty}^\pm(\s)\right)'$ and $\left(w_\beta(\s)-w_{\beta,\infty}^\pm(\s)\right)''$ are
exponentially decaying. Here we use that $H_1$ is of class $C^3$ (even of class $C^4$). These estimates are used to estimate
$L(w_\beta-w_{\beta,\infty}^{\pm})$ below.

As $ L w_{\beta,\infty}^\pm-\alpha \psi'(w_{\beta,\infty}^\pm)=0$ by Theorem~\ref{thm:def of H1}, we find that
\begin{multline}
  \label{eq: exp dec}
  Lw_\beta-\alpha \psi'(w_\beta)=L(w_\beta-w_{\beta,\infty}^{\pm})
  -\alpha\left(\psi'(w_\beta)-\psi'(w^{\pm}_{\beta,\infty})\right)
  \\=L(w_\beta-w_{\beta,\infty}^{\pm})
  -\alpha\int_0^1\psi''\left( w_{\beta,\infty}^{\pm}+\si(w_\beta-w_{\beta,\infty}^{\pm})\right)
  d\si\cdot(w_\beta-w_{\beta,\infty}^{\pm}),
\end{multline}
and both terms on the right are also exponentially decaying as $\s\rightarrow \pm\infty$ by the exponential bound on~\eqref{eq:
  the diff} just established. Note that, if $\epsilon>0$ is small enough, $|w_\beta|,|w_{\beta,\infty}^{\pm}| \ge\epsilon$ for all
$|\s| \geq C_0\epsilon$ and thus
$\left|\psi''\left( w_{\beta,\infty}^{\pm}+\si(w_\beta-w_{\beta,\infty}^{\pm})\right)\right|\leq \Const\,\epsilon$ uniformly in
$\beta$ and $\si\in[0,1]$.  By continuity of $\psi'$ and the other expressions involved,
$Lw_\beta-\alpha \psi'(w_\beta) \in L^\infty((-C_0 \epsilon, C_0\epsilon))$ (remember that
$w_\beta\in C^2(\RR\backslash\{0\})\cap W^{2,\infty}(\RR)$).  These arguments prove the first part of claim~\ref{it:behaviour1},
$Lw_\beta-\alpha \psi'(w_\beta)\in \XX$, the dependence in $\beta\in[-1,1]$ being continuous (decreasing $|\nu|$ further if
necessary).

To establish that this expression is $C^1([-1,1], E^\nu_0(\RR\backslash[-2,2], \RR))$ as a function of $\beta$, we give an argument in
three steps.

Step 1. First note that for fixed small $\epsilon>0$, the map
\begin{equation*}
  \beta \rightarrow Y_\beta:= Lw_\beta-\alpha\psi' (w_\beta)
\end{equation*}
is $C^1$ in $\beta$ if $Y_\beta$ is restricted to any bounded interval $(\s_0,\s_1)\subset \RR\backslash[-2,2]$ and the target
space is endowed with the norm of $L^\infty((\s_0,\s_1))$.  Indeed, $w_\beta$ is obtained from $\beta$ by composition of $C^2$
maps (in fact $C^4$ maps) in $\beta$, $u_p$, $u_o$, $u'_p$, $u'_o$ and the change of variables
$x\rightarrow \widetilde x=2\pi \s/(\widetilde \sP(\beta)k_0)$.  Moreover,
\begin{multline*}
  \frac{d}{d\beta} (w_\beta(x))
  =  \partial_1H_1(\cdot,\cdot)\left(Bu_o(\ts)
    +\left\{u_p'(\ts)+B\beta u'_o(\ts)\right\}
    \ts(-\widetilde \sP'(\beta)/\widetilde \sP(\beta))\right)
  \\+\partial_2H_1(\cdot,\cdot)\Big(Bu'_o(\ts)
  +\{u''_p(\ts)+B\beta u''_o(\ts)\}
  \ts(-\widetilde \sP'(\beta)/\widetilde \sP(\beta))\Big).
\end{multline*}
In the less regular term $L\frac d {d\beta}w_\beta$, $u_p''(\ts)$ is differentiated twice with respect to $\s$, and thus $u_p''''$
arises.  As $u_p$ is in $C^4(\RR\setminus[-1.5,1.5])$, the $C^1$ regularity in $\beta$ follows (for small enough $\epsilon>0$).

Step 2. Hence it remains to check that $\beta \rightarrow Y_\beta$ is $C^1$ if the target space is endowed with the norm of
$E^\nu_0(\RR\backslash[-2,2], \RR)$. Due to the previous argument, the claim is proved if $Y_\beta$ and $\frac d{d\beta} Y_\beta$
are bounded in $E^{\widetilde \nu}_0(\RR\backslash[-2,2], \RR)$, uniformly in $\beta$, where $\widetilde \nu<\nu<0$. We observe
that the estimate on $Y_\beta$ in~\eqref{eq: exp dec} is uniform in $\beta$, so it remains to analyse the derivative with respect
to $\beta$ in the final step; note that ~\eqref{eq: exp dec} and Step 3 establish these two properties for $\nu$ (taking $|\nu|$
small enough in Step 3 below), so the claim follows for slightly smaller $|\nu|$.

Step 3. We recall that $\widetilde \sP'(\beta)\rightarrow 0$ uniformly in $\beta$ as $\epsilon\rightarrow 0$ (see
Theorem~\ref{thm:def of H1}). We thus obtain in analogy to~\eqref{eq: the diff} and~\eqref{eq: exp dec} that
$\frac d {d\beta}\left(w_\beta-w_{\beta,\infty}^\pm\right)$ and $\frac{d}{d\beta}\left(Lw_\beta-\alpha\psi'(w_\beta)\right)$ are
exponentially decaying as $\s\rightarrow \pm\infty$.  Here we use that $H_1$ is of class $C^4$.

This shows that, after decreasing $|\nu|$ if necessary,
\begin{equation*}
  \frac{d}{d\beta}\left(L w_\beta-\alpha \psi'(w_{\beta})\right)\in E^\nu_0(\RR\setminus[-2,2], \RR)
\end{equation*}
with uniform bounds in $\beta\in[-1,1]$ and small $\epsilon>0$.

We move on to claim~\ref{it:behaviour2}, the transversality relation.  Remember that $H_1(\z,\xx)$ tends to $\z$,
$\partial_1 H_1(\z,\xx)$ tends to $1$ and $\partial_2 H_1(\z,\xx)$ tends to $0$ as $\epsilon\rightarrow 0$ by Theorem~\ref{thm:def
  of H1}. Using these properties and again the fact that $\widetilde \sP'(\beta)\rightarrow 0$ uniformly in $\beta$ as
$\epsilon\rightarrow 0$, we obtain
\begin{align*}
  \frac{d}{d\beta}(L w_\beta&-\alpha \psi'(w_\beta))=(L-\alpha\psi''(w_\beta)I)
  \circ 
  \\ \Big[&
  \partial_1H_1(\cdot,\cdot)\left(Bu_o(\ts)
  +\left\{u_p'(\ts)+B\beta u'_o(\ts)\right\}
  \ts(-\widetilde \sP'(\beta)/\widetilde \sP(\beta))\right)
  \\&+\partial_2H_1(\cdot,\cdot)\Big(Bu'_o(\ts)
  +\{u''_p(\ts)+B\beta u''_o(\ts)\}
  \ts(-\widetilde \sP'(\beta)/\widetilde \sP(\beta))\Big)
  \Big]
\end{align*}
converges to $BLu_o$ in $L^\infty_{\loc}(\RR\backslash[-2,2])$ uniformly in $\beta\in[-1,1]$ as $\epsilon\rightarrow 0$.  As a
consequence, the map $\beta\rightarrow \int_{|\s|> 2} \Big(Lw_\beta-\alpha \psi'(w_\beta)\Big)\sin(\k_0\s)d\s$ is $C^1$ and
\begin{equation*}
   \frac{d}{d\beta}\int_{|\s|> 2}\Big(Lw_\beta-\alpha     \psi'(w_\beta)\Big)\sin(\k_0\s)d\s
\end{equation*}
converges to $\int_{|\s|>2} L(Bu_o)\sin(\k_0\s)d\s$ as $\epsilon\rightarrow 0$, uniformly in $\beta\in[-1,1]$.

To deal with the values of $\s$ in $[-2,2]$, we write the integral on $[-2,2]$ differently in claim~\ref{it:behaviour2}:
\begin{align*}
  &\int_{-2}^2\Big(Lw_\beta-\alpha  \psi'(w_\beta)\Big)\sin(\k_0\s)d\s
  \\&
      =c^2w_\beta'(2)\sin(\k_02)-c^2w_\beta'(-2)\sin(-\k_02)
      -\int_{-2}^2c^2k_0w_\beta'\cos(\k_0\s)d\s
  \\&\qquad
      {}+\int_{-2}^2\Big(-\Delta_D w_\beta+\alpha w_\beta-\alpha  \psi'(w_\beta)\Big)\sin(\k_0\s)d\s.
\end{align*}
By symmetry, we can consider $\s\in[0,2]$ only, so the previous expression equals
\begin{multline*}
  2c^2w_\beta'(2)\sin(\k_02)
    -2\int_{0}^2c^2k_0w_\beta'\cos(\k_0\s)d\s
  \\{}+2\int_{0}^2\Big(-\Delta_D w_\beta+\alpha w_\beta-\alpha  \psi'(w_\beta)\Big)\sin(\k_0\s)d\s, 
\end{multline*}
which clearly is $C^1$ in $\beta\in[-1,1]$; its derivative is
\begin{multline*}
  2c^2\frac d {d\beta}(w_\beta'(2))\sin(\k_02)
    -2\int_{0}^2c^2k_0\frac d {d\beta}(w_\beta')\cos(\k_0\s)d\s
  \\+2\int_{0}^2\Big(-\Delta_D\frac d {d\beta} (w_\beta)+\alpha \frac d {d\beta}(w_\beta)-\alpha\psi''(w_\beta)  
      \frac d {d\beta}(w_\beta)\Big)\sin(\k_0\s)d\s
\end{multline*}
with
\begin{multline*}
  \frac d {d\beta}\left( w_\beta(x) \right)
    = \partial_1H_1(\cdot,\cdot)\left(Bu_o(\ts)
    +\left\{u_p'(\ts)+B\beta u'_o(\ts)\right\}
    \ts(-\widetilde \sP'(\beta)/\widetilde \sP(\beta))\right)
  \\+\partial_2H_1(\cdot,\cdot)\Big(Bu'_o(\ts)
      +\{u''_p(\ts)+B\beta u''_o(\ts)\}
      \ts(-\widetilde \sP'(\beta)/\widetilde \sP(\beta))\Big),
\end{multline*}
$x\in(0,2]$ and $\widetilde \s=2\pi \s/(\widetilde \sP(\beta)k_0)$. From this, we get a formula for
$\frac d {d\beta}(w_\beta'(x))$ for $x\in(0,2]$ and $\beta\in[-1,1]$.  As already observed, $H_1(\z,\xx)$ tends to $\z$ (as
$\epsilon\rightarrow 0$), $\partial_1 H_1(\z,\xx)$ tends to $1$, $\partial_2 H_1(\z,\xx)$ tends to $0$ and
$\widetilde \sP'(\beta)$ tends to $0$ uniformly in $\beta$. Remember the hypothesis~\eqref{eq:psi2prime},
$|\psi''_\epsilon(u)|\leq 2 \epsilon^{-1}$ for all $|u|<\epsilon$, which leads to $|\psi''_\epsilon(u)|< \Const \epsilon^{-1}$ for
all $u\in \RR$ and for some constant $\Const>0$.  We thus get
\begin{equation*}
  \int_{0}^{C_0\epsilon}
  \left|\alpha\psi''(w_\beta(\s))\sin(\k_0\s)\right|d\s
  \leq \Const\alpha \epsilon^{-1} \k_0\int_{0}^{C_0\epsilon}
  |\s|d\s\rightarrow 0
\end{equation*} 
and by~\eqref{eq: cdn on psi}
\begin{equation*}
  \int_{C_0\epsilon}^2
  \left|\alpha\psi''(w_\beta(\s))\sin(\k_0\s)\right|d\s
  \leq \Const\alpha \epsilon \int_{C_0\epsilon}^2d\s
  \rightarrow 0
\end{equation*} 
as $\epsilon\to 0$, uniformly in $\beta\in[-1,1]$. Hence, as $\epsilon\to 0$,
\begin{multline*}
  \frac d {d\beta}\Big(c^2w_\beta'(2)\sin(\k_02)-c^2w_\beta'(-2)\sin(-\k_02)
    -\int_{-2}^2c^2k_0w_\beta'\cos(\k_0\s)d\s
  \\ +\int_{-2}^2(-\Delta_D w_\beta+\alpha w_\beta-\alpha  \psi'(w_\beta))\sin(\k_0\s)d\s\Big)
\end{multline*}
converges to
\begin{multline*}
  c^2Bu_o'(2)\sin(\k_02)-c^2Bu_o'(-2)\sin(-\k_02)
    -\int_{-2}^2c^2k_0Bu_o'\cos(\k_0\s)d\s
  \\
      +\int_{-2}^2\Big(-\Delta_D Bu_o+\alpha Bu_o\Big)\sin(\k_0\s)d\s
      =\int_{-2}^2  L(Bu_o)\sin(\k_0\s)d\s,
\end{multline*}
uniformly in $\beta\in[-1,1]$. 

Finally, note that Proposition~\ref{prop: orthogonality} yields
\begin{equation*}
  \int_\RR  BLu_o(\s) \sin(\k_0\s) d\s=B(-2c\k_0^2+2)\neq 0 .
\end{equation*}
Reducing $B>0$ if needed, we have proved~\eqref{eq: K_0} and the remark on $K_0$ that follows it.

It remains to show claim~\ref{it:behaviour3}, that is,
\begin{equation*}
  \lim_{(\beta,\epsilon)\rightarrow 0}\sup\left\{
    e^{|\nu \s|}\left|Lw_{\beta}(\s)-\alpha\psi'(w_\beta(\s))\right|:
    \s\in \RR, |w_{\beta}(\s)|\geq \epsilon\right\}=0.
\end{equation*}
In~\eqref{eq: exp dec}, we have
\begin{equation*}
  \left|\chi_{\{|w_\beta|\geq \epsilon\}}
    \psi''(w_{\beta,\infty}^{\pm}
    +\si\{w_\beta-w_{\beta,\infty}^{\pm}\})\right|\leq \Const\epsilon
\end{equation*}
and thus
\begin{multline*}
  \left\|e^{|\nu\s|}\chi_{\{|w_\beta|\geq \epsilon\}}\cdot(w_{\beta}-w_{\beta,\infty}^{\pm})
    \int_0^1\psi''\left(w_{\beta,\infty}^{\pm}
    +\si\{w_\beta-w_{\beta,\infty}^{\pm}\}\right)d\si
  \right\|_{L^\infty(\RR)}
  \\\leq \Const \epsilon\rightarrow 0
\end{multline*}
as $\epsilon\rightarrow 0$, by taking a smaller $|\nu|$ if needed (see~\eqref{eq: the diff}).  Moreover, there exists some
$\widetilde \nu<0$ provided $|\nu|$ is taken small enough, so that for $\pm$ in~\eqref{eq: exp dec} replaced by $+$ (respectively
$-$),
\begin{multline*}
  e^{|\nu \s|}L(w_\beta-w_{\beta,\infty}^{\pm})
  =e^{|\nu \s|}L\Big(
  H_1\Big(u_p(\ts)+ B\beta u_o(\ts),
  u'_{p}(\ts)+ B\beta u_o'(\ts)\Big)
  \\-H_1\Big(u_{p,\infty}^{\pm}(\ts)+ B\beta u_{o,\infty}^{\pm}(\ts),
  {u_{p,\infty}^{\pm}}'(\ts)+ B\beta {u_{o,\infty}^{\pm}}'
  (\ts)\Big)\Big)
\end{multline*}
has its absolute value bounded from above by $\Const e^{-\widetilde \nu |\s|} $ on $(0,\infty)$ (respectively $(-\infty,0)$) and
converges uniformly on every bounded subset of $(0,\infty)$ (respectively $(-\infty,0)$) to
\begin{equation*}
  e^{|\nu\s|}L(u_p-u_{p,\infty}^{\pm})=
  e^{|\nu\s|}(L u_p-(\pm\alpha))=0
\end{equation*}
as $(\beta,\epsilon)\rightarrow 0$; see~\eqref{eq:kz} and~\eqref{eq:jodel}.  Claim~\ref{it:behaviour3} follows. \qed

Motivated by the spaces $E^\nu_\m(X)$, we define the solution spaces for the ``corrector'' $r$ used in Step 4 of the argument, as
outlined at the end of Section~\ref{sec:Setting-main-result}. For $\nu<0$, let
\begin{equation}
  \label{eq:E0}
  E^\nu_{0,odd}(\RR):=\{r\in C(\RR):\text{ $r$ is odd and } e^{|\nu \s|}r(\s)\in
  L^{\infty}(\RR)\}
\end{equation} 
and
\begin{equation}
 \label{eq:E1}
  E^\nu_{1,odd}(\RR):=\{r\in E^\nu_{0,odd}(\RR)\cap C^{1}(\RR):
  e^{|\nu \s|}r'(\s)\in L^{\infty}(\RR)\}.
\end{equation} 

\begin{lemma}
  \label{lemma about B(0,rho)}
  If $B,\rho,\epsilon>0$ are chosen small enough, then for all $r$ in the ball $\overline{B(0,\rho)}\subset E^\nu_{1,odd}(\RR)$
  \begin{multline}
    \label{eq: lemma part 1}
    \sup_{\beta\in[-1,1]}\int_{\RR}\left|\alpha
      \left(\psi''(w_\beta)-\psi''(w_\beta-r)\right)
      \frac d {d\beta}w_\beta\sin(\k_0\s)\right|d\s
    \\
    \leq\frac 1 2   \inf_{\beta\in[-1,1]}
    \left|  \frac{d}{d\beta}\int_\RR(
       (L w_\beta
      -\alpha\psi'(w_\beta))
      \sin(\k_0\s) d\s\right|
    :=\frac 1 2 \kk_0>0
  \end{multline}
  and 
  \begin{multline}
    \label{eq: lemma part 2}
    \int_{\RR}\left|(c^2w_0''-\Delta_D w_0+\alpha w_0
      -\alpha \psi'(w_0-r))\sin(\k_0\s)\right|d\s
    \\
    \leq\frac 1 2  \inf_{\beta\in[-1,1]}
    \left |\int_{\RR}\frac{d}{d\beta}
      (L w_\beta
      -\alpha\psi'(w_\beta))\sin(\k_0\s) d\s\right|=\frac 1 2 \kk_0.
  \end{multline}
\end{lemma}

\proof Choose $B,\rho>0$ small enough and $C_0>0$ large enough (in a way that is independent of small $\epsilon>0$) so that
$|w_{\beta}-r|\geq \epsilon$ on $\RR\backslash(-C_0\epsilon,C_0\epsilon)$ for all $\beta\in[-1,1]$, $r\in \overline{B(0,\rho)}$
and all small $\epsilon$. We set $I_1 := \{\s : |w_\beta(\s) |< \epsilon\}\subset (-C_0\epsilon,C_0\epsilon)$,
$I_2 := \{\s : |w_\beta(\s) -r|< \epsilon\}\subset (-C_0\epsilon,C_0\epsilon)$ and $I_3 := \RR \setminus (I_1 \cup I_2)$. If
$B,\rho,\epsilon>0$ are chosen small enough, for all $|\beta|\leq 1$ we get from~\eqref{eq:psi2prime} for $j = 1, 2$
\begin{multline*}
  \int_{I_j}
  \left|\psi''(w_\beta)-\psi''(w_\beta-r)\right|
  \,\left|\frac d{d\beta}w_\beta\right|\,
  |\sin(\k_0\s)|d\s
  \\\leq
  \int_{I_j} \Const\epsilon^{-1}
  \,\left|\frac d{d\beta}w_\beta\right|\k_0|\s| d\s
  \stackrel{\eqref{eq: w_beta}}\leq \Const \epsilon,
\end{multline*}
and for their complement from~\eqref{eq: cdn on psi}
\begin{align*}
  &
    \int_{I_3}
    \left|\psi''(w_\beta)-\psi''(w_\beta-r)\right|
    \,\left|\frac d{d\beta}w_\beta\right|\,|\sin(\k_0\s)|d\s
  \\&=
      \int_{I_3} 
      \left|\int_0^1\psi'''(w_\beta -\sigma r)rd\sigma\right|
      \,\left|\frac d{d\beta}w_\beta\right|\,|\sin(\k_0\s)|d\s
  \\&\leq
      \int_{I_3} \Const\epsilon|r|
      \,\left|\frac d{d\beta}w_\beta\right|\,|\sin(\k_0\s)|d\s
      \leq \Const \epsilon,
\end{align*}
and thus combined
\begin{equation*}
  \int_{\RR}\left|\alpha\left(\psi''(w_\beta)-\psi''(w_\beta-r)
    \right)\frac d {d\beta}w_\beta\sin(\k_0\s)\right|d\s
  \leq \Const \epsilon.
\end{equation*}
As $K_0>0$ uniformly for small $\epsilon>0$, by~\eqref{eq: K_0} and the remark on $K_0$ following it, this proves~\eqref{eq: lemma
  part 1}.

Also, $Lw_{0}-\alpha \psi'(w_0)\in E^\nu_0(\RR\backslash[-1,1],\RR)$ with uniform bounds for small enough $\epsilon$ (by taking a
smaller $|\nu|$ if necessary, see the comments after~\eqref{eq: the diff} and~\eqref{eq: exp dec}). Further,
$Lw_{0}-\alpha \psi'(w_0)$ is bounded in $L^\infty((-1,1))$, uniformly for small $\epsilon$.  Moreover, $w_{0}$ converges to $u_p$
in $W^{2,\infty}_{\loc}(\RR)$ as $\epsilon$ tends to $0$ (see~\eqref{eq: w_beta}) and $\psi'(w_{0}(\s))$ converges to
$\text{sgn}(\s)$ for $x\in\RR\setminus\{0\}$ as $\epsilon$ tends to $0$.  Thus
\begin{equation}
  \label{eq: tends to 0}
  \int_{\RR}\left|(L w_0 
    -\alpha \psi'(w_0))
    \sin(\k_0\s)\right|d\s \to 0 \text{ as } \epsilon \to 0
\end{equation}
since $Lu_p(\s)- \alpha\text{sgn}(\s)=0$; we recall that $u_p$ is the solution for the special case of a piecewise quadratic
potential~\cite{Kreiner2011a}.  Moreover,
\begin{align*}
  &\int_{\RR}|\psi'(w_0)-\psi'(w_0-r)|\,|\sin(\k_0\s)|d\s
  \\&=
      \int_{\RR}\left|\int_0^1\psi''(w_0-\si r)rd\si\right|\,|\sin(\k_0\s)|d\s
  \\&\leq
      \int_{\RR}\left(\Const\epsilon^{-1}\chi_{\{|x|<C_0\epsilon\}}
      +\Const\epsilon\chi_{\{|x|\geq C_0\epsilon\}}\right)|r|\,|\sin(\k_0\s)|d\s
  \\&\leq
      \Const\int_{\RR}\epsilon^{-1}\chi_{\{|x|<C_0\epsilon\}}|r|\,\k_0|\s|d\s
      +\Const\int_{\RR}\epsilon\chi_{\{|x|\geq C_0\epsilon\}}|r|\,|\sin(\k_0\s)|d\s
      \leq \Const \epsilon
\end{align*}
and, as a consequence (see~\eqref{eq: tends to 0}),
\begin{equation}
  \label{eq: h(0)}
  \int_{\RR}\left|(
    Lw_0 - \alpha \psi'(w_0-r))
    \sin(\k_0\s)\right|d\s\rightarrow 0
\end{equation}
uniformly in $r\in \overline{B(0,\rho)}$ as $\epsilon\rightarrow 0$.  As $\kk_0>0$, this proves~\eqref{eq: lemma part 2}.  \qed

\section{Existence of a heteroclinic connection}
\label{sec:Exist-heter-conn}

In this section, we employ a fixed point argument to prove the existence of a ``corrector'' $r$ required in Step 4 (introduced at
the end of Section~\ref{sec:Setting-main-result}). We recall the definition of the solution spaces $E_{0,odd}^{\nu}(\RR)$ and
$E_{1,odd}^{\nu}(\RR)$ in~\eqref{eq:E0} and~\eqref{eq:E1}. In addition, we introduce the following Banach space.  Let
$G^\nu_{0}(\RR)$ be the Banach space of functions $f\in L^{2}(\RR)$ such that
\begin{equation} 
\label{eq:def G}
  ||f||_{G^\nu_{0}(\RR)}:=
  ||e^{-\nu |\cdot|}f||_{L^2(\RR)}<\infty.
\end{equation}
 
Given $\nu<0$, we would like to find $r\in E_{1,odd}^{\nu}(\RR)\cap H^2(\RR)$ and $\beta\in [-1,1]$ such that $w_\beta-r$ is a
solution to equation~\eqref{eq:r},
\begin{equation*}
  c^2\Big(w_\beta-r\Big)''
  -\Delta_D \Big(w_\beta -r\Big)
  \\+\alpha\Big(w_\beta-r\Big)
  -\alpha\psi'\Big(w_\beta-r\Big) 
  =0.
\end{equation*}
We shall apply Proposition~\ref{prop: on the exponential decay}, the remark following it and Proposition~\ref{prop: on the
  exponential decay bis}.  They address the solution the equation $Lr=Q$ for $r\in E_{1}^{\nu}(\RR)\cap H^2(\RR)$, where $Q$ is in
various spaces of decaying functions and satisfies
\begin{equation*}
  \int_\RR Q(\s)\sin(\k_0 \s)d\s=\int_\RR Q(\s)\cos(\k_0 \s)d\s=0.
\end{equation*}
In Proposition~\ref{prop: on the exponential decay}, $Q$ belongs to $E_0^{\nu}(\RR)$ (that is, $Q$ is continuous and the function
$e^{|\nu \cdot|}Q$ is bounded); in the remark, $Q$ belongs to $F_0^{\nu}(\RR)$ (that is, $Q$ and the function $e^{|\nu \cdot|}Q$
are in $L^\infty(\RR)$) and in Proposition~\ref{prop: on the exponential decay bis}, $Q$ belongs to $G_0^{\nu}(\RR)$.  Since in
the present section $Q$ is odd, only the condition $\int_\RR Q(\s)\sin(\k_0 \s)d\s$ has to be dealt with and
$r\in E_{1,odd}^{\nu}(\RR)\cap H^2(\RR)$ with $\nu<0$ and $|\nu|$ sufficiently small.

\begin{lemma}
  \label{lemma: compact}
  The map
  \begin{equation}
    \label{eq:Gamma}
    (r,\beta)\rightarrow \Gamma(r,\beta) :=
    c^2w_\beta''-\Delta_D w_\beta+\alpha w_\beta
    -\alpha \psi'\big(w_\beta-r\big)-\alpha \psi''(w_\beta)r
  \end{equation}
  is well defined as a map $E^\nu_{1,odd}(\RR)\times [-1,1] \to \XX$ and is continuous.  Moreover, $\Gamma$ is compact (that is,
  every bounded sequence in the domain of $\Gamma$ is mapped by $\Gamma$ into a sequence that has a convergent subsequence).  The
  map $(r,\beta)\rightarrow \int_\RR\Gamma(r,\beta)\sin(k_0x)dx$ is $C^1$ on $E^\nu_{1,odd}(\RR)\times [-1,1]$.
\end{lemma}

\proof The map $\beta \to \Gamma (0, \beta) = c^2w_\beta''-\Delta_D w_\beta+\alpha w_\beta -\alpha \psi'\big(w_\beta)$ is
well-defined and continuous by claim~\ref{it:behaviour1} of Proposition~\ref{prop: behaviour}. To prove the lemma, we first
investigate the difference in the nonlinear terms of the last expression and the one in~\eqref{eq:Gamma}. By the fundamental
theorem of calculus and Fubini's theorem, we obtain that
\begin{align*}
  &\psi'(w_\beta)- \psi'\big(w_\beta-r\big)-\psi''(w_\beta)r
  =\int_0^1\left(\psi''(w_\beta-\si r)-\psi''(w_\beta)\right)rd\si
  \\&=-\int_0^1\left(\int_0^\si \psi'''(w_\beta-\tilde \si r)r^2d\tilde \si\right) d\si
 =-\int_0^1 \left( \int_{\tilde\si}^1 \psi'''(w_\beta-\tilde \si r)r^2d\si\right) d\tilde\si 
  \\&=-\int_0^1 (1-\tilde \si)\psi'''(w_\beta-\tilde \si r)r^2d\tilde \si.
\end{align*}
Setting $e^{|\nu \s|}r(\s)=:\widetilde r(\s)$, we thus have to show that the map
\begin{equation*}
  W^{1,\infty}(\RR)\times[-1,1]
  \ni(\widetilde r,\beta)\rightarrow  -e^{-|\nu \s|}
  \int_0^1 (1-\sigma)\psi'''(w_\beta-\sigma e^{-|\nu \s|} \widetilde r)
  \widetilde r^2d\sigma
\end{equation*}
with range included in $L^\infty(\RR)$ is $C^1$ and compact. The first two properties are immediate (thanks to the weight
$e^{-|\nu\cdot|}$ in front of the integral), and compactness is a consequence of the Arzel\`a-Ascoli theorem (using the uniform
continuity of $\psi'''$ on compact sets and the weight $e^{-|\nu \cdot|}$ in front of the integral). To conclude, we recall that
the map $\beta \to \int_\RR\Gamma (0, \beta)\sin(k_0x)dx$ is $C^1$ by claim~\ref{it:behaviour2} of Proposition~\ref{prop:
  behaviour}.  \qed

\begin{proposition}
  \label{prop: def beta}
  Let $\rho>0$ be small enough (see Lemma~\ref{lemma about B(0,rho)}).  
  Let
 \begin{equation*}
    h(r,\beta):=
    \int_\RR\left(c^2w_\beta''-\Delta_D w_\beta+\alpha w_\beta
    -\alpha\psi'\big(w_\beta-r) \right)\sin(\k_0 \s)d\s=0
  \end{equation*}  
  be a map $E^\nu_{1,odd}(\RR)\times [-1,1] \to \RR$.  For fixed $r$ in $\overline{B(0,\rho)}$, the equation
  \begin{equation*}
    h(r,\beta)=0
  \end{equation*}
  has a unique solution $\beta=\beta(r)$, which is a $C^1$-function of $r$.  Moreover, $\beta(r)$ tends to $0$ uniformly in
  $r\in \overline{B(0,\rho)}$ as $\epsilon\rightarrow 0$.
\end{proposition}

Observe that since $\Gamma(r,\beta)$ and $\alpha \psi''(w_\beta)r$ are integrable over $\RR$, so is the integrand in the previous definition of $h$.

\proof We shall check at the end of the proof that the map $h$ is of class $C^1$. Let $r$ in $\overline{B(0,\rho)}$. The
proposition is a consequence of the fact that, for all $\beta\in [-1,1]$,
\begin{align*}
  |\tfrac{\partial}{\partial\beta}h(r,\beta)|&=\left|\frac{\partial}{\partial\beta}
    \int_\RR\left(c^2w_\beta''-\Delta_D w_\beta+\alpha w_\beta
    -\alpha\psi'\big(w_\beta-r\big) \right)\sin(\k_0\s)d\s\right|
  \\&\stackrel{
    \eqref{eq: lemma part 1}}{\geq }
  \frac 1 2  \inf_{ \widetilde \beta\in[-1,1]}
  \left|\frac{d}{d\beta}\int_\RR\Big(c^2w_{\widetilde\beta}''
    -\Delta_Dw_{\widetilde\beta}+\alpha w_{\widetilde \beta}
    -\alpha\psi'(w_{\widetilde \beta})\Big)\sin(\k_0\s)d\s\right| \\
  &=\frac 1 2 \kk_0,
\end{align*}
which implies $\inf_{\beta\in[-1,1]}|\tfrac{\partial}{\partial\beta}h(r,\beta)| \stackrel{\eqref{eq: lemma part 2}}\geq |h(r,0)|$.
In turn, this implies $h(r,\beta)=0$ for some $\beta=\beta(r)\in[-1,1]$, as desired.  To this behalf we argue by contradiction and
assume for definiteness that
$\inf_{\widetilde \beta\in[-1,1]}\tfrac{\partial h}{\partial\beta}(r,\widetilde \beta)\geq h(r,0)> 0$, so we may set
\begin{equation*}
  b=b(r):=-\frac{h(r,0)}{\inf_{\widetilde \beta\in[-1,1]}\tfrac{\partial h}{\partial\beta}(r,\widetilde \beta) }\in[-1,0).
\end{equation*} 
Then by the intermediate value theorem there exists a $\widetilde \beta\in (b,0)$ such that
\begin{equation*}
  h(b)=h(b)-h(0)+h(0)=\frac{\partial h}{\partial \beta}(\widetilde \beta)(b-0)+h(0)
  \leq -h(0)+h(0)=0
\end{equation*}
(for simplicity, we have omitted the depencence on $r$).  Therefore $h(r,\beta)=0$ for some $\beta=\beta(r)\in [b(r),0)$.

More generally, $|\beta(r)|\leq |h(r,0)|/\inf_{\widetilde \beta\in[-1,1]} |\frac{\partial h}{\partial\beta}(r,\widetilde \beta)|$
because the infimum is positive (see above).  Hence $\beta(r)$ tends to $0$ uniformly in $r\in \overline{B(0,\rho)}$ as
$\epsilon\rightarrow 0$, see~\eqref{eq: h(0)}.  Further, considering the differentiability of $h(r,\beta)$ with respect to $r$,
the $C^1$-dependence of $\beta(r)$ on $r$ is a consequence of the implicit function theorem.

To check that the map $h$ is of class $C^1$, it is sufficient to check that the map
$(r,\beta)\rightarrow \int_\RR\psi''(w_\beta) \sin(\k_0 \s)r\,d\s:=\widetilde h(r,\beta)$ is of class $C^1$ (by Lemma~\ref{lemma:
  compact}).  For fixed $\beta\in[-1,1]$, $\widetilde h(r,\beta)$ is linear in $r$ and, as a linear functional, is bounded with
norm bounded above by $||e^{\nu|\cdot|}\psi''(w_\beta)\sin(\k_0\s)||_{L^1(\RR)}$:
\begin{equation*}
  |\widetilde h(r,\beta)|\leq ||e^{\nu|\cdot|}\psi''(w_\beta)\sin(\k_0\s)||_{L^1(\RR)}||r||_{E^\nu_{1,odd}(\RR)}~\text{ for 
    all }~r\in E^\nu_{1,odd}(\RR).
\end{equation*}
Hence, for fixed $(r,\beta)$, $\frac{\partial }{\partial r}\widetilde h(r,\beta): E^\nu_{1,odd}(\RR)\rightarrow \RR$ is the same
bounded linear functional
\begin{equation*}
  \rho\rightarrow \int_\RR\psi''(w_\beta) \sin(\k_0 \s)\rho\,d\s~\text{ for all }~\rho\in E^\nu_{1,odd}(\RR),
\end{equation*}
with $||\frac{\partial }{\partial r}\widetilde h(r,\beta)|| \leq ||e^{\nu|\cdot|}\psi''(w_\beta)\sin(\k_0\s)||_{L^1(\RR)}$.
Moreover, $\frac{\partial }{\partial r}\widetilde h(r,\beta)$ depends only on $\beta\in[-1,1]$ and is continuous, because the map
$\beta\rightarrow e^{\nu|\cdot|}\psi''(w_\beta)\sin(\k_0\s)\in L^1(\RR)$ is continuous.  Also, for any fixed
$r\in E^\nu_{1,odd}(\RR)$, $\widetilde h(r,\beta)$ is of class $C^1$ in $\beta\in[-1,1]$, with continuous partial derivative in
$(r,\beta)$.  Indeed, its partial derivative at $\beta\in[-1,1]$ is
\begin{equation*}
  \frac{\partial }{\partial \beta}\widetilde h(r,\beta)=\int_\RR\psi'''(w_\beta)
  \left(\frac{d}{d\beta}w_{\beta}\right)\sin(\k_0 \s)r\,d\s.
\end{equation*}
As the map $\beta\rightarrow e^{\nu|\cdot|}\psi'''(w_\beta) \left(\frac{d}{d\beta}w_{\beta}\right)\sin(\k_0\s)\in L^1(\RR)$ is
continuous too, $\frac{\partial }{\partial \beta}\widetilde h(r,\beta)$ is continuous with respect to $(r,\beta)$. The linear
factor $\frac{-2\pi\widetilde \sP'(\beta) \s}{\widetilde \sP^2(\beta)k_0}$ that appears in $\frac{d}{d\beta}w_{\beta}$ is
controlled by the factor $e^{\nu|x|}$ (as $\nu<0$).

\subsubsection*{Proof of Theorem~\protect{\ref{theo:main}}}
\label{sec:Proof-Theor}

The result of this section so far can be formulated as follows. The problem can be written as $c^2r''-\Delta_D r+\alpha r=Q$ with
$Q:=\Gamma(r,\beta(r))+\alpha\psi''(w_{\beta(r)})r\in\XX$ odd and $\int_{\RR} Q(\s)\sin(\k_0 \s)d\s=0$ by the definition of
$\beta(r)$ in Proposition~\ref{prop: def beta}.

Let $r=L^{-1}Q$ be given by Proposition~\ref{prop: on the exponential decay} (and Remark) in Appendix~\ref{sec:Tools-from-Fourier}
applied to $Q$ defined above, so that our problem can be rewritten as
\begin{equation}
  \label{eq:r2}
  r=L^{-1}Q=
  L^{-1}\Big(\Gamma(r,\beta(r))+\alpha\psi''(w_{\beta(r)})r\Big).
\end{equation}

For $\widetilde \beta\in[-1,1]$, let $\delta(r,\widetilde \beta)\in \RR$ be such that
\begin{equation*}
  \int_\RR\Big(\alpha\psi''(w_{\widetilde \beta})r
  -\delta(r,\widetilde \beta)Lu_o\Big)\sin(\k_0\s)d\s=0;
\end{equation*}
Proposition~\ref{prop: orthogonality} shows that
$\displaystyle \delta(r,\widetilde \beta)=\frac{ \int_\RR\alpha\psi''(w_{\widetilde \beta})r\sin(k_0\s)d\s}{-2c^2\k_0+2}$. Then
\begin{equation*}
  \int_\RR\Big(\Gamma(r,\widetilde \beta)
  +\delta(r,\widetilde \beta)Lu_o\Big)(\s)\sin(\k_0\s)d\s=0,
\end{equation*}
\begin{equation}
  \label{eq: estimate delta}
\delta(r,\widetilde \beta)=-\frac{ \int_\RR\Gamma(r,\widetilde \beta)\sin(k_0\s)d\s}{-2c^2\k_0+2}
\end{equation}
and our problem becomes
\begin{equation*}
	r=L^{-1}\Big(\Gamma(r,\beta(r))
	+\delta(r,\beta(r))Lu_o\Big)+
	L^{-1}\Big(\alpha\psi''(w_{\beta(r)})r
	-\delta(r,\beta(r))Lu_o\Big).
\end{equation*}

Choose $C_0>0$ large enough and $\rho>0$ small enough so that $|w_{\beta} -r|\geq \epsilon$ on
$\RR\backslash(-C_0\epsilon,C_0\epsilon)$ for all $r\in \overline{B(0,\rho)}$, $\beta\in[-1,1]$ and all small $\epsilon$. 

Define $\delta_1(r,\widetilde\beta)$ and $\delta_2(r,\widetilde\beta,R)$ for $R\geq 1$ by
\begin{equation*}
  \int_{\RR}\Big(\alpha\psi''(w_{\widetilde\beta})
  \chi_{\{|x|<C_0 \epsilon\}}r
  -\delta_1(r,\widetilde \beta)Lu_o\Big)(\s)\sin(\k_0\s)d\s=0
\end{equation*}
and
\begin{equation*}
  \int_{\RR}\Big(\alpha\psi''(w_{\widetilde \beta})
  \chi_{\{C_0 \epsilon\leq |x|<R\}}r
  -\delta_2(r,\widetilde \beta,R)Lu_o\Big)(\s)\sin(\k_0\s)d\s=0,
\end{equation*}
that is, by Proposition~\ref{prop: orthogonality},
\begin{equation*}
  \delta_1(r,\widetilde \beta)
  =\frac{\int_{\RR}\alpha\psi''(w_{\widetilde\beta})
    \chi_{\{|x|<C_0 \epsilon\}}\sin(\k_0\s)rd\s}{-2c^2\k_0+2}
\end{equation*}
and
\begin{equation*}
  \delta_2(r,\widetilde \beta,R)
  =\frac{\int_{\RR}\alpha\psi''(w_{\widetilde\beta})
  \chi_{\{C_0\epsilon\leq |x|<R\}}\sin(\k_0\s)rd\s}{-2c^2\k_0+2}\,.
\end{equation*}
Thanks to $|\sin(\k_0\s)|\leq |\k_0\s|$,
\begin{equation}
  \label{eq: estimate on r}
  |\psi''(w_{\widetilde\beta})|<\Const\epsilon^{-1}~~\text{and}~~
  |r(\s)|\leq \Const|\s|\cdot
  ||r'||_{L^\infty((-C_0\epsilon,C_0\epsilon))}
\end{equation}
on $(-C_0\epsilon,C_0\epsilon)$ (because $r(0)=0$), we get
$\delta_1(r,\widetilde\beta)=O\Big(\epsilon^2 ||r'||_{L^\infty((-C_0\epsilon,C_0\epsilon))}\Big)$.  Also
$|\psi''(w_{\widetilde\beta})|<\Const\,\epsilon$ on $\RR\backslash (-C_0\epsilon,C_0\epsilon)$ and
$\delta_2(r,\widetilde \beta,R)=O\Big(\epsilon||r||_{L^1(\RR)}\Big)$, uniformly in $\widetilde \beta\in[-1,1]$.  The maps
$\delta$, $\delta_1$ and $\delta_2$ are clearly linear in $r$ and, moreover, continuous because of the continuity of the map
\begin{equation*}
  \widetilde \beta\rightarrow \psi''(w_{\widetilde \beta}) \in L^\infty_{\loc}(\RR)
\end{equation*}
(this means that the map $\widetilde \beta \rightarrow \psi''(w_{\widetilde \beta})\chi_{\{|x|<R\}} \in L^\infty(\RR)$ is
continuous for every finite $R>0$).  Furthermore,
\begin{multline}
  \label{eq: on psi''}
  \|\alpha\psi''(w_{\widetilde\beta}
  \,+\sigma r)\chi_{\{C_0\epsilon\leq|\s|<R\}}r\|_{F^\nu_0(\RR)}
\\  \leq \|\alpha\psi''(w_{\widetilde\beta}
  \,+\sigma r)\chi_{\{C_0\epsilon\leq|\s|<\infty\}}r\|_{F^\nu_0(\RR)}
  \leq \Const\epsilon  \|r\|_{E^\nu_0(\RR)}
\end{multline}
for all $\sigma\in [-1,0]$, and
\begin{equation}
  \label{eq: on psi'' bis}
  \|\alpha\psi''(w_{\widetilde\beta})  \chi_{\{|\s|<C_0\epsilon\}}r\|_{G^\nu_{0}(\RR)}
  \leq \Const\epsilon^{1/2} ||r'||_{L^\infty((-C_0\epsilon,C_0\epsilon))}\,.
\end{equation}
See~\eqref{eq: estimate on r}.  By Proposition~\ref{prop: on the exponential decay},
\begin{equation*}
  \Big\|L^{-1}\Big(\alpha\psi''(w_{\widetilde\beta})
  \chi_{\{C_0\epsilon\leq|\s|<R\}}r
  -\delta_2(r,\widetilde \beta,R)Lu_o\Big)\Big\|
  _{E^\nu_{1,odd}(\RR)}
  \leq \Const\epsilon  \|r\|_{E^\nu_0(\RR)}
\end{equation*}
and, by Proposition~\ref{prop: on the exponential decay bis},
\begin{equation*}
  \Big\|L^{-1}\Big(\alpha\psi''(w_{\widetilde\beta})
  \chi_{\{|x|<C_0\epsilon\}}r-\delta_1(r,\widetilde\beta)Lu_o\Big)\Big\|
  _{E^\nu_{1,odd}(\RR)}\leq \epsilon^{1/2} O(||r'||
  _{L^\infty((-C_0\epsilon,C_0\epsilon))})
\end{equation*}
uniformly in $\widetilde\beta\in[-1,1]$. Hence the linear map
\begin{equation}
  \label{eq: linear map to be inverted}
  r\rightarrow r-
  L^{-1}\Big(\alpha\psi''(w_{\widetilde \beta})\chi(|\s|<R)r
  -(\delta_1(r,\widetilde\beta)+\delta_2(r,\widetilde \beta,R))Lu_o\Big)
\end{equation}
is invertible if $\epsilon$ is small enough.  Let us denote the inverse by
$\Xi_{\widetilde\beta}\colon E^\nu_{1,odd}(\RR) \rightarrow E^\nu_{1,odd}(\RR)$, which is continuous in $\widetilde \beta$ when
the operator norm is considered. To check the continuity of $\Xi_{\widetilde \beta}$ in $\widetilde \beta$, observe that the
linear operator
\begin{equation}
  \label{eq: linear operator}
  E^\nu_{1,odd}(\RR)\, \ni \,r\rightarrow 
  \alpha\psi''(w_{\widetilde \beta})\chi_{\{|\s|<R\}}r\, \in \, F^\nu_{0,odd}(\RR)
\end{equation}
is bounded, its operator norm being bounded above by $||\alpha\psi''(w_{\widetilde \beta})\chi_{\{|\s|<R\}}||_{L^{\infty}(\RR)}$.
As the map $\widetilde \beta\rightarrow \alpha\psi''(w_{\widetilde \beta})\chi_{\{|\s|<R\}}\in L^{\infty}(\RR)$ is continuous, so
is the linear operator~\eqref{eq: linear operator} with respect to the operator norm.  The linear functionals 
\begin{equation*}
  r\rightarrow \delta_1(r,\widetilde\beta)
  ~~\text{ and }~~r\rightarrow \delta_2(r,\widetilde \beta,R)
\end{equation*}
on $E^\nu_{1,odd}(\RR)$ are also continuous in $\widetilde\beta$ with respect to the usual norm for linear functionals (see the
end of the proof of Proposition \ref{prop: def beta} for a similar functional).  As a consequence, the linear operator
\begin{equation*}
  E^\nu_{1,odd}(\RR)\, \ni \,r\rightarrow 
  \alpha\psi''(w_{\widetilde \beta})\chi_{\{|\s|<R\}}r\
  -(\delta_1(r,\widetilde\beta)+\delta_2(r,\widetilde \beta,R))Lu_o
  \, \in \, F^\nu_{0,odd}(\RR)
\end{equation*}
is bounded and continuous in $\widetilde \beta$ with respect to the operator norm, and thus the bounded linear operator~\eqref{eq:
  linear map to be inverted} (defined on $E^\nu_{1,odd}(\RR)$) is continuous in $\widetilde \beta$ with respect to the operator
norm.  The continuity of $\Xi_{\widetilde \beta}$ with respect to $\widetilde \beta$ results from the fact that, in a Banach
space, the map that sends an invertible bounded operator with bounded inverse to its inverse is continuous with respect to the
operator norm.  In what follows, the continuity of $\Xi_{\widetilde \beta}$ is only needed for every fixed $R\geq 1$.

On the other hand, the map
\begin{equation*}
  r\rightarrow 
  L^{-1}\Big(\Gamma(r,\beta(r))+\delta(r,\beta(r))Lu_o\Big)
\end{equation*}
is completely continuous on $\overline{B(0,\rho)}$ (that is, continuous and compact); see Lemma~\ref{lemma: compact}.  Therefore
\begin{equation*}
  r\rightarrow \Xi_{\beta(r)}\Big(L^{-1}\Big(\Gamma(r,\beta(r))
  +\delta(r,\beta(r))Lu_o\Big)\Big)
\end{equation*} 
is completely continuous, too. For $\epsilon>0$ small enough, it sends $\overline{B(0,\rho)}$ into $B(0,\rho)$.  To see this, we
refer to the remark after Proposition~\ref{prop: on the exponential decay}, \eqref{eq: estimate delta} and use
\begin{align*}
  &\|\Gamma(r,\beta(r))\chi_{\{|\s|\geq C_0\epsilon\}}\|_{F^\nu_0(\RR)}
  \\ &\leq
       \|\left(L w_{\beta(r)}
       -\alpha \psi'\big(w_{\beta(r)}-r\big)\right)
       \chi_{\{|\s|\geq C_0\epsilon\}}\|_{F^\nu_0(\RR)}
  \\
  &{}\qquad+
    \|\alpha \psi''(w_{\beta(r)})r
    \chi_{\{|\s|\geq C_0\epsilon\}}\|_{F^\nu_0(\RR)}
  \\ &\stackrel{\eqref{eq: on psi''}} \leq
       \|\left(L w_{\beta(r)}
       -\alpha \psi'\big(w_{\beta(r)}-r\big)\right)
       \chi_{\{|\s|\geq C_0\epsilon\}}\|_{F^\nu_0(\RR)}
       +\Const\epsilon||r||_{E^\nu_{0,odd}(\RR)}
  \\ &\leq 
       \Const\epsilon||r||_{E^\nu_{0,odd}(\RR)}
       +
       \|\left(L w_{\beta(r)}
       -\alpha \psi'\big(w_{\beta(r)}\big)\right)
       \chi_{\{|\s|\geq C_0\epsilon\}}\|_{F^\nu_0(\RR)}
  \\ &\qquad{}+
       \left\|\alpha \int_{-1}^0 \psi''(w_{\beta(r)}+\si r)rd\si
       \chi_{\{|\s|\geq C_0\epsilon\}}\right\|_{F^\nu_0(\RR)}
       \stackrel{\eqref{eq: on psi''}}\rightarrow 0
\end{align*}
uniformly in $r\in \overline{B(0,\rho)}$ as $\epsilon$ tends to $0$, thanks to the third part of Proposition~\ref{prop: behaviour}
and the fact that $\beta(r)$ tends uniformly to $0$ as $\epsilon$ tends to $0$ (see Proposition~\ref{prop: def beta}).  We use
also Proposition~\ref{prop: on the exponential decay bis}, that
$\Gamma(r,\beta(r)) \chi_{\{|\s|<C_0\epsilon\}}\in F^\nu_0(\RR)\cap G^{\nu}_0(\RR)$ and
\begin{align*}
  &\left\|\Gamma(r,\beta(r)) \chi_{\{|\s|<C_0\epsilon\}}\right\|_{G^\nu_{0}(\RR)}
  \\ 
  &\leq
    \left\|
    \left(L w_{\beta(r)}
    -\alpha \psi'\big(w_{\beta(r)}-r\big)
    \right)
    \chi_{\{|\s|<C_0\epsilon\}}\right\|_{G^\nu_{0}(\RR)}
  \\
  &\qquad{}+\left\|
    \alpha \psi''(w_{\beta(r)})r
    \chi_{\{|\s|<C_0\epsilon\}}\right\|_{G^\nu_{0}(\RR)}
  \\&\stackrel{\eqref{eq: on psi'' bis}} \leq
      \left\|
      \left(L w_{\beta(r)}
      -\alpha \psi'\big(w_{\beta(r)}-r\big)
      \right)
      \chi_{\{|\s|<C_0\epsilon\}}\right\|_{G^\nu_{0}(\RR)}
      +\Const\epsilon^{1/2}||r||_{E^\nu_{1,odd}(\RR)}
  \\&\leq 
      \Const\epsilon^{1/2}
      +\Const\epsilon^{1/2}||r||_{E^\nu_{1,odd}(\RR)}
      \rightarrow 0
\end{align*}
uniformly in $r\in \overline{B(0,\rho)}$ as $\epsilon$ tends to $0$, since $|\psi'|$ is uniformly bounded
(see~\eqref{eq:psi2prime} and~\eqref{eq:psiprimeoutside}) and so is $|(L w_{\beta(r)})(x)|$.

Thus the Schauder fixed point theorem gives a solution $r=r_R\in \overline{B(0,\rho)}$ to the equation
\begin{equation*}
  r=\Xi_{\beta(r)}\Big(L^{-1}\Big(\Gamma(r,\beta(r))
  +\delta(r,\beta(r))Lu_o\Big)\Big),
\end{equation*} 
which can be written as 
\begin{align*}
  &  r=L^{-1}\Big(\Gamma(r,\beta(r))
    +\delta(r,\beta(r))Lu_o\Big)
  \\&\qquad{}
      +L^{-1}\Big(
      \alpha\psi''(w_{\beta(r)})\chi(|\s|<R)r
      -\delta_1(r,\beta(r))Lu_0-\delta_2(r,\beta(r),R)Lu_o\Big)
  \\& \qquad{}
      \in B(0,\rho)\cap H^2_{odd}(\RR)
\end{align*}
(see Proposition~\ref{prop: on the exponential decay bis}). As the estimates above are uniform in $R\geq 1$, we also get that
$||r||_{H^2(\RR)}=||r_R||_{H^2(\RR)}$ is uniformly bounded in $R$. Hence
\begin{multline*}
  Lr=
  c^2w_{\beta(r)}''-\Delta_D w_{\beta(r)}+\alpha w_{\beta(r)}
  -\alpha \psi'\big(w_{\beta(r)}-r\big)
  \\
  -\alpha \psi''(w_{\beta(r)})\chi(|\s|\geq R)r
  +\Big(\delta(r,\beta(r))-\delta_1(r,\beta(r))-\delta_2(r,\beta(r),R)\Big)Lu_o
\end{multline*}
with $r=r_R\in B(0,\rho)\cap H^2_{odd}(\RR)$ and a uniform bound on $||r_R||_{H^2(\RR)}$. Therefore there exists a sequence
$R_n\rightarrow \infty$, $r_\infty\in \overline B(0,\rho)\cap H^2_{odd}(\RR)$ and $\beta_\infty\in[-1,1]$ such that
$r_{_{R_n}}\rightarrow r_\infty$ weakly in $H^2_{odd}(\RR)$ and $\beta(r_{_{R_n}})\rightarrow \beta_\infty$.  Taking limits in the
above equation, we deduce that
\begin{align*}
  Lr_\infty=
  c^2w_{\beta_\infty}''-\Delta_D w_{\beta_\infty}+\alpha w_{\beta_\infty}
  -\alpha \psi'\big(w_{\beta_\infty}-r_\infty\big).
\end{align*}
We have used that
\begin{equation*}
  \delta_1(r,\widetilde \beta)+\delta_2(r,\widetilde \beta,R)\rightarrow
  \delta(r,\widetilde \beta)
\end{equation*}
as $R\rightarrow \infty$, uniformly in $\widetilde \beta\in[-1,1]$ and $r\in \overline B(0,\rho)$.  \qed

\appendix
\renewcommand*{\thesection}{\Alph{section}}
\makeatletter
\renewcommand\section{\@startsection {section}{1}{\z@}%
           {18\p@ \@plus 6\p@ \@minus 3\p@}%
           {9\p@ \@plus 6\p@ \@minus 3\p@}%
           {\normalsize\bfseries\boldmath \noindent Appendix }}
\makeatother
\section{Tools from Fourier analysis}
\label{sec:Tools-from-Fourier}

Although the Banach spaces considered in this article are real, it is convenient to consider them as complex when working with the
Fourier transform. We begin with a straightforward but useful generalisation of results in~\cite{Iooss2000a}.  For $\nu\in \RR$,
$\m\in\{0,1,2,\ldots\}$ and a Banach space $X$, we recall that $E^\nu_\m(X)$ is the Banach space of functions $f\in \Cm(\RR,X)$
equipped with the norm~\eqref{eq:Eknu},
\begin{equation*}
  ||f||_{E^\nu_\m(X)}:=\max_{0\leq j\leq \m}
  ||e^{-\nu |\cdot|} f^{(j)}||_{L^\infty(\RR,X)} < \infty.
\end{equation*}
In the case $X=\RR$, this means $f\in E^\nu_\m(\RR)$ if and only if $f\in \Cm(\RR)$ satisfies
\begin{equation*}
  \max_{0\leq j\leq \m}\sup_{\s\in \RR}e^{-\nu |\s|}|f^{(j)}(\s)|<\infty.
\end{equation*}
In the case $X=C([-1,1])$, $f\in E^\nu_\m(C([-1,1]))$ can be identified with the continuous mapping
$(\s,\v)\mapsto \tilde f(\s,\v)\in \RR$, where $\tilde f(\s,\v)$ is the value at $\v\in[-1,1]$ of $f(\s)\in C([-1,1])$; then
$f\in E^\nu_\m(C([-1,1]))$ if and only if each $\partial_1^j\tilde f$ exists and belongs to $C(\RR\times[-1,1])$ for
$0\leq j\leq \m$, and
\begin{equation*}
  \max_{0\leq j\leq \m}\sup_{(\s,\v)\in \RR\times[-1,1]}e^{-\nu |\s|}
  |\partial_1^j\tilde f(\s,\v)|<\infty.
\end{equation*}
In the case $X=C^1([-1,1])$, $f\in E_0^\nu(C^1([-1,1]))$ can be identified with $\tilde f \in C(\RR\times[-1,1])$ such that in
addition to the requirements for $X=C([-1,1])$, also $\partial_2\tilde f$ exists and belongs to $C(\RR\times[-1,1])$, and
\begin{equation*}
  \max_{j\in\{0,1\}}\sup_{(\s,\v)\in \RR\times[-1,1]}e^{-\nu |\s|}
  |\partial_2^j\tilde f(\s,\v)|<\infty.
\end{equation*}

\begin{proposition}
  \label{prop: adaptation}
  Let $p_0>0$ and the measurable map $(\k,\v)\mapsto\widehat H(\k,\v)\in \CC$ be defined on its domain
  \begin{equation*}
    \{ (\k,\v)\in \CC\times[-1,1]: \Im\, \k\in(-p_0,p_0)\}.
  \end{equation*} 
  We assume that, for each $\v\in[-1,1]$, the map $\k\mapsto\widehat H(\k,\v)$ is analytic in the strip
  $\{\k\in \CC:\Im\, \k\in(-p_0,p_0)\}$ and, for all $\delta\in(0,p_0)$, $(1+|\k|)|\widehat H(\k,\v)|$ is bounded in
  $\{(\k,\v)\in \CC\times[-1,1]: \Im\, \k\in[-\delta,\delta]\}$.

  Then, for every $\v\in[-1,1]$, $\widehat H(\cdot,\v)\colon\RR\rightarrow\CC$ is the Fourier transform of some
  $H(\cdot,\v)\in L^2(\RR)$,
  \begin{equation*}
    \widehat H(k,\v)=\int_{\RR}e^{-ik\s}H(\s,\v)d\s ; 
  \end{equation*}
  the map $(\s,\v)\rightarrow H(\s,\v)$ being measurable on $\RR\times[-1,1]$. Moreover, for each $\nu\in(-p_0,p_0)$, the linear
  map $f\rightarrow H\star f$ is well defined from $E_0^\nu(C([-1,1]))$ into itself and is uniformly bounded if $\nu$ is restricted
  to be in any compact subset of $(-p_0,p_0)$.  Here, the convolution is taken with respect to the real variable $\s$ only,
  \begin{equation*}
    (H\star f)(\s,\v)=\int_\RR H(\s-\u,\v)f(\u,\v)d\u.
  \end{equation*}
  
  If in addition $|\k^2\widehat H(\k,\v)|$ is bounded in $\{(\k,\v)\in \CC\times[-1,1]: \Im\, \k\in[-\delta,\delta]\}$ for all
  $\delta\in (0,p_0)$, then the map $f\rightarrow H\star f$ is well defined from $E_0^\nu(C([-1,1]))$ into $E_1^\nu(C([-1,1]))$
  and is uniformly bounded if $\nu$ is restricted to be in any compact subset of $(-p_0,p_0)$.
\end{proposition}

\begin{proof}
  We remark that if $\widehat H(\k,\v)$ and $f(\s,\v)$ are both independent of $\v$, this proposition (and the following proof) is
  essentially~\cite[Lemma 3]{Iooss2000a}.  Let $0<\delta<p_0$.  We have $(1+|\k|^2)^{1/2}|\widehat H(\k,\v)|\leq \Const$ on
  $\{(\k,\v)\in \RR\times[-1,1]:\Im\, \k\in[-\delta,\delta]\}$ and $\widehat H(\cdot,\v)\in L^2(\RR)$.  Hence, for every
  $\v\in[-1,1]$, $\widehat H(\cdot,\v)$ is the Fourier transform of some $H(\cdot,\v)\in L^2(\RR)$, the map
  $(\s,\v)\mapsto H(\s,\v)$ being measurable. Moreover, by the Cauchy theorem on contour integrals in the complex plane,
\begin{multline*}
  e^{\delta \s}H(\s,\v)
  =\frac{1}{2\pi}e^{\delta \s}\int_\RR e^{i\s\k}\widehat H(\k,\v)d\k
  \\=\frac{1}{2\pi}e^{\delta \s}\int_\RR e^{i\s(i\delta+\k)}\widehat H(i\delta+\k,\v)d\k
  =\frac{1}{2\pi}\int_\RR e^{i\s\k}\widehat H(i\delta+\k,\v)d\k,
\end{multline*}
and thus, by Plancherel,
\begin{equation*}
  ||e^{\delta \cdot}H(\cdot,\v)||_{L^2(\RR)}
  =\frac{1}{\sqrt{2\pi}}||\widehat H(i\delta+\cdot,\v)||_{L^2(\RR)}.
\end{equation*}
The same estimate with $\delta$ replaced by $-\delta$ gives
\begin{equation}
  \label{eq: est on H}
  \sup_{\v\in[-1,1]}||\,e^{\delta |\cdot|}H(\cdot,\v)\,||_{L^2(\RR)}<\infty.
\end{equation}
Let $|\nu|<\delta$, $\v\in[-1,1]$ and convolutions be only with respect to $\s$. As in~\cite{Iooss2000a}, we get for all
$f\in E^\nu_0(C([-1,1]))$
\begin{multline*}
  \sup_{(\s,\v)\in \RR\times[-1,1]}
  e^{-\nu|\s|}\left|\int_{\RR}H(\s-\u,\v)f(\u,\v)d\u\right|
  \\\leq||f||_{E^\nu_0(C[-1,1])}\sup_{(\s,\v)\in \RR\times[-1,1]}
  \int_\RR e^{-\nu|\s|+\nu |\u|-\delta|\s-\u|}
  \left|e^{\delta|\s-\u|}H(\s-\u,\v)\right|d\u
  \\\leq||f||_{E^\nu_0(C[-1,1])}
  \sup_{\v\in[-1,1]}||\,e^{\delta |\cdot|}H(\cdot,\v)\,||_{L^2(\RR)}
  \sup_{\s\in \RR}\left(\int_\RR e^{2\nu(|\u|-|\s|)-2\delta|\s-\u|}d\u\right)^{1/2}
  \\\leq||f||_{E^\nu_0(C[-1,1])}
  \sup_{\v\in[-1,1]}||\,e^{\delta |\cdot|}H(\cdot,\v)\,||_{L^2(\RR)}
  \sup_{\s\in \RR}\left(\int_\RR e^{2|\nu|(|\u-\s|)-2\delta|\s-\u|}d\u\right)^{1/2}
  \\=||f||_{E^\nu_0(C[-1,1])}
  \sup_{\v\in[-1,1]}||\,e^{\delta |\cdot|}H(\cdot,\v)\,||_{L^2(\RR)}
  (\delta-|\nu|)^{-1/2}\,.
\end{multline*}

If in addition $|\k^2\widehat H(\k,\v)|$ is bounded in $\{(\k,\v)\in \CC\times[-1,1]: \Im\, \k\in[-\delta,\delta]\}$, then we can
apply the previous argument to $i\k\widehat H=\widehat{\partial_\s H}$ instead of $\widehat H$, noting
$\partial_\s(H\star f)=(\partial_\s H)\star f$.
\end{proof}

Recall the dispersion function $D(\k)=-c^2\k^2+2(1-\cos \k)+\alpha$ and let
\begin{equation}
  \label{eq: p_0} 
  p_0:=\inf\{|\Im\, \k|:D(\k)=0,\, \Im\, \k\neq 0\}>0.
\end{equation}
By Lemma 1 in~\cite{Iooss2000a}, $p_0>0$.

Observe that, in contrast with Proposition~\ref{prop: adaptation}, $\nu$ is required to be negative in the following statement.
\begin{proposition}
  \label{prop: on the exponential decay}
  Let $\nu\in(-p_0,0)$.  If $Q\in E^\nu_{0}(\RR)$ satisfies
  \begin{equation}
    \label{eq: int=0} 
    \int_\RR Q(\s)\sin(\k_0 \s)d\s=\int_\RR Q(\s)\cos(\k_0 \s)d\s=0,
  \end{equation}
  then, for all $c \leq 1$ close enough to $1$, there exists a unique function $r\in E^\nu_{1}(\RR)\cap H^2(\RR)$ such that
  $Lr=Q$.  Moreover, the linear map $Q\rightarrow r$ is bounded as a map $ E^\nu_{0}(\RR) \to E^\nu_{1}(\RR)$.
\end{proposition}

\begin{proof}
Let us formally define the function $r$ by its Fourier representation $\widehat r(\k):=\widehat Q(\k)/D(\k)$.  As $D$
vanishes on $\RR$ exactly at $\pm \k_0$ with non-vanishing derivative $D'(\pm \k_0)=\pm(-2c^2\k_0+2)\neq 0$, we can
define the function
\begin{equation*}
  f(k) := \frac{-2c^2\k_0+2}{2\k_0}(\k^2-\k_0^2),
\end{equation*}
which also vanishes exactly at $\pm \k_0$ and satisfies there $f'(\pm \k_0) = D'(\pm \k_0)$.  Thus, we can write
\begin{equation*}
  \frac 1 {D(\k)}=
  \frac1{f(\k)}+\widehat H(\k)
\end{equation*}
with a remainder function $\widehat H(\k)$.  Clearly $\widehat H(\k)$ is analytic in the strip
$\{\k\in \CC: \Im\, \k\in(-p_0,p_0)\}$.  As $|k^2/D(k)|$ is bounded in
\begin{equation*}
  \{\k\in \CC: \Im\, \k\in(-p_0,p_0),\, |D(k)|>1\},
\end{equation*}
we know that$|\k^2\widehat H(\k)|$ is bounded on the strip $\{\k\in \CC: \Im\, \k\in(-\delta,\delta)\}$ for all
$\delta\in (0,p_0)$.  Thus Proposition~\ref{prop: adaptation} applied to the case when $\widehat H$ and $f$ do not depend on the
second variable $\v$, the map $Q\mapsto H \star Q$ is well defined and bounded from $E^{\nu}_{0}(\RR)$ to
$E^{\nu}_{1}(\RR)$. Moreover, $H\star Q$ is clearly in $H^2(\RR)$ as $Q$ is assumed to decay exponentially.  Note that
$H\in H^1(\RR)$.

On the other hand, we ignore $f$ for the moment and notice that the function $\frac{1}{\k_0^2-\k^2}\widehat Q(\k)$ is related to
the Fourier transform of the solution $r_0(\s)$ of the equation $L_0r_0=r_0''+k_0^2r_0=Q$. The variation of constants formula and
\eqref{eq: int=0} give
\begin{equation*}
  r_0(\s)=\frac{1}{\k_0}\int_{-\infty}^{\s}\sin(\k_0(\s-\u))Q(\u)d\u
  =\frac 1{\k_0}\int_\s^{\infty}\sin(\k_0(\u-\s))Q(\u)d\u
\end{equation*}
with
\begin{equation*}
  r_0'(\s)=\int_{-\infty}^{\s}\cos(\k_0(\s-\u))Q(\u)d\u
  =-\int_\s^{\infty}\cos(\k_0(\u-\s))Q(\u)d\u.
\end{equation*}
It easily follows that $r_0\in E^{\nu}_{1}(\RR)$, $r_0\in H^2(\RR)$ and that the map $Q\mapsto r_0$ is bounded as map
$E^\nu_{0}(\RR) \to E^\nu_{1}(\RR)$ and $r\in H^2(\RR)$.

Combining the two previous steps and noting that the solution $r$ can, by definition of $\widehat H$, be written as
$r=-\frac{2\k_0}{-2c^2\k_0+2}r_0+H\star Q$, we have proved the claim.
\end{proof}

We remark that in the previous Proposition, the hypothesis that $Q$ is continuous is actually not used; it suffices to assume that
$Q\in L^\infty(\RR)$ and $e^{|\nu\cdot|}Q\in L^\infty(\RR)$.

In the same way, one gets the following theorem, in which the assumption $Q\in E^{\nu}_0(\RR)$ is replaced by
$e^{|\nu \cdot|}Q\in L^2(\RR)$.
\begin{proposition}
  \label{prop: on the exponential decay bis}
  Suppose that $\nu\in(-p_0,0)$.  If $Q\in L^2(\RR)$, $e^{|\nu \cdot|}Q\in L^2(\RR)$ and
  \begin{equation*}
    \int_\RR Q(\s)\sin(\k_0 \s)d\s=\int_\RR Q(\s)\cos(\k_0 \s)d\s=0,
  \end{equation*}
  then, for all $c \leq 1$ close enough to $1$, there exists a unique function $r\in E^\nu_{1}(\RR)\cap H^2(\RR)$ such that
  $Lr
  =Q$.  Moreover,
  \begin{equation*}
    ||r||_{E^\nu_{1}(\RR)}\leq \Const||e^{|\nu \cdot|}Q||_{L^2(\RR)}.
  \end{equation*}
\end{proposition}

\begin{proof}
  With $H$ as in the proof of Proposition~\ref{prop: on the exponential decay}, let us check that
\begin{equation*}
  ||H\star Q||_{E^\nu_{0}(\RR)}\leq \Const||e^{|\nu \cdot|}Q||_{L^2(\RR)}
\end{equation*} 
for all negative $-\delta<\nu<0$. Indeed,
\begin{align*}
  &e^{|\nu \s|}|(H\star Q)(\s)|
  =e^{|\nu \s|}\left|\int_\RR H(\s-\u)Q(\u)d\u\right|
  \\&\leq \left(\sup_{\s,\u\in\RR} e^{|\nu \s|-\delta |\s-\u|-|\nu \u|}
  \right)
  \int_\RR e^{\delta|\s-\u|}|H(\s-\u)|e^{|\nu \u|}|Q(\u)|d\u
  \\&\leq ||e^{\delta |\cdot|}H||_{L^2(\RR)}||e^{|\nu\cdot|}Q||_{L^2(\RR)},
\end{align*}
where we used $\left|\,|\nu t|-|\nu\tau|\,\right|\leq \delta|t-\tau|$ (see also~\eqref{eq: est on H} for an $H$ independent of
$\v$).  Similarly one can prove that
\begin{equation*}
  ||H'\star Q||_{E^\nu_{0}(\RR)}\leq ||e^{|\nu \cdot|}Q||_{L^2(\RR)}.
\end{equation*}
Finally, for the solution $r_0$ of $L_0r_0=Q$, the variation of constants formula implies that
$r_0\in E^{\nu}_{1}(\RR)$ and $
||r_0||_{E^\nu_{1}(\RR)}\leq \Const||e^{|\nu \cdot|}Q||_{L^2(\RR)}.  $
\end{proof}

We also use the following result, which is proved in~\cite[Proposition A.2]{Buffoni2017a}.
\begin{proposition}
  \label{prop: orthogonality} 
  If $u_o\in C(\RR)$ satisfies \eqref{eq: bound u_o} and $c>\k_0^{-1/2}$, then
  \begin{equation*}
    \int_{\RR}\sin(\k_0\s )(c^2u_o''-\Delta_D u_o+\alpha u_o)d\s=-2c^2\k_0+2<0.
  \end{equation*}
\end{proposition}

We repeat the proof for the sake of completeness.
\begin{proof}
  Two integrations by parts and the identity $L\sin(\k_0\s)=0$ give
  \begin{align*}
    &\lim_{R\rightarrow \infty}\int_{-R}^{R}\sin(\k_0\s )
       (c^2u_o''-\Delta_D u_o+\alpha u_o)d\s
    \\&=\lim_{R\rightarrow \infty}\int_{-R}^{R}\Big(
         c^2\frac{d^2}{d\s^2}\sin(\k_0\s )-\Delta_D\sin(k_0\s)
         +\alpha \sin(\k_0\s)\Big)u_o\, d\s
    \\&\qquad{}+\lim_{\s\rightarrow \infty}c^2
         \{\sin(\k_0 \s)u'_0(\s)-\k_0\cos(\k_0 \s)u_o(\s)
    \\&\qquad\qquad
         - \sin(-\k_0 \s)u'_0(-\s)+\k_0\cos(-\k_0 \s)u_o(-\s)\}
    \\&\qquad{}-\lim_{R\rightarrow \infty}\left(\int_{-R+1}^{R+1}-\int_{-R}^{R}\right)\sin(\k_0(\s-1))u_o(\s)d\s
    \\&\qquad{}-\lim_{R\rightarrow \infty}\left(\int_{-R-1}^{R-1}-\int_{-R}^{R}\right)\sin(\k_0(\s+1))u_o(\s)d\s \\
    &\stackrel{\eqref{eq: bound u_o}}=
      \lim_{\s\rightarrow \infty}c^2
      \{-\k_0\sin^2(\k_0 \s)-\k_0\cos^2(\k_0 \s)
      -\k_0 \sin^2(\k_0 \s)-\k_0\cos^2(\k_0 \s)\}
    \\&\qquad{}-\lim_{R\rightarrow \infty}\left[
        \int_{R}^{R+1}\sin(\k_0(\s-1))\cos(\k_0\s)d\s \right. 
        \\&\qquad\qquad{}+\left.\int_{-R}^{-R+1}\sin(\k_0(\s-1))\cos(\k_0\s)d\s\right]
    \\&\qquad{}+\lim_{R\rightarrow \infty}\left[
        \int_{-R-1}^{-R}\sin(\k_0(\s+1))\cos(\k_0\s)d\s \right.
         \\&\qquad{}+ \left.
        \int_{R-1}^{R}\sin(\k_0(\s+1))\cos(\k_0\s)d\s
        \right]
    \\&=-2c^2\k_0
        +\lim_{R\rightarrow \infty}\left(\int_{R-1}^{R+1}\cos^2(\k_0\s)d\s
        +
        \int_{-R-1}^{-R+1}\cos^2(\k_0\s)d\s
        \right)
    \\&
        =-2c^2\k_0+2<0.
  \end{align*}
\end{proof}

\section{Application of centre manifold theory}
\label{sec:Appl-cent-manif}

Let $Y$ be any Banach space such that $\DD\subset Y\subset \HH$ with continuous embeddings (but not necessarily dense).  To check
the hypotheses in~\cite{Vanderbauwhede1992a}, it suffices to check that, for all $\nu\in[0,p_0)$ and all
$G=(G_0,G_1,G_2)\in E^\nu_0(Q_hY)$, there exists a unique $U=(\z,\xx,\Z)\in E^{\nu}_0(Q_h\DD)\cap C^1(\RR,Q_h\HH)$ such that
\begin{equation}
  \label{eq: equation for U}
  \partial_t U=L_{\gamma,\tau}U+G.
\end{equation}
Here the constant $p_0$ can be as in~\eqref{eq: p_0}, or any smaller positive constant.  Writing~\eqref{eq: equation for U} as
$U=KG$, we also need to check (as required in \cite{Vanderbauwhede1992a}) that $K\in \sL(E^\nu_0(Q_hY),E^\nu_0(Q_h \DD))$ and
\begin{equation*}
  ||K||_\nu\leq \widetilde\gamma(\nu)
\end{equation*}
for some continuous function $\widetilde \gamma \colon[0,p_0)\rightarrow[0,\infty)$.

In~\cite{Iooss2000a}, this is proved when $G(t)$ is of the particular form $G(t)=Q_h(0,G_1(t),0)$ and this is sufficient for the
proof of~\cite{Vanderbauwhede1992a} to work. However, to fulfil the hypotheses of the statement of~\cite{Vanderbauwhede1992a},
this should be proved at least for the more general case $G(t)\in Q_h\DD$. For completeness, let us check this hypothesis for all
$G\in E^\nu_0(Q_h \DD)$, that is, $Y=\DD$, following the same method as in~\cite{Iooss2000a}.  Its validity is an obvious
consequence of Theorem~\ref{thm: solution U} below.

Let us assume that $\nu\in(-p_0,p_0)$ and let $G=(G_0,G_1,G_2)\in E^\nu_0(Q_h \DD)$.  The condition $G(t)\in Q_h \DD$ is
equivalent to the set of four conditions (see Lemma 2 in~\cite{Iooss2000a}): $G_2(t,\cdot)\in C^1([-1,1])$, $G_0(t)=G_2(t,0)$
\begin{equation}
  \label{eq: G_0}
  \k_0G_0(t)=\gamma\tau^2\int_0^1\sin(\k_0(1-\v))[G_2(t,\v)+G_2(t,-\v)]d\v
\end{equation}
and
\begin{equation}
  \label{eq: G_1}
  G_1(t)=\gamma\tau^2\int_0^1\cos(\k_0(1-\v))[G_2(t,\v)-G_2(t,-\v)]d\v.
\end{equation}
Properties~\eqref{eq: G_0} and~\eqref{eq: G_1} together are equivalent to $G(t)\in Q_h\HH$.  For
$G_2(t)=G_2(t,\cdot)\in C^1([-1,1])$, $G_2(t)$ is the last component of some $G(t)\in Q_h \DD=\DD\cap Q_h\HH$ if and only if
\begin{equation}
  \label{eq: necessary cdn G2}
  \k_0G_2(t,0)=\gamma\tau^2\int_0^1\sin(\k_0(1-\v))[G_2(t,\v)+G_2(t,-\v)]d\v.
\end{equation}

\begin{theorem}
  \label{thm: solution U}
  Let the constant $p_0>0$ be as in~\eqref{eq: p_0}.
  \begin{enumerate}
  \item For every $\nu\in(-p_0,p_0)$, consider the bounded linear map with bounded inverse that sends
    $G_2\in E_0^\nu(C^1([-1,1]))$ satisfying~\eqref{eq: necessary cdn G2} to $G=(G_0,G_1,G_2)\in E^\nu_0(Q_h \DD)$ with $G_0$ and
    $G_1$ given by~\eqref{eq: G_0} and~\eqref{eq: G_1}.  There exists a bounded linear map
    \begin{equation*}
      \widetilde K:G_2\mapsto U\in E^\nu_0(\DD)
    \end{equation*}
    defined for $G_2$ as above such that
    \begin{equation*}
      U\in C^1(\RR,\HH) ~\text{ and }~\partial_t U=L_{\gamma,\tau}U+G.
    \end{equation*}
  \item Moreover, $U\in E^{\nu}_0(Q_h\DD) \cap C^1(\RR,Q_h\HH)$.
  \item The solution $U$ is unique in $E^{\nu}_0(Q_h\DD)\cap C^1(\RR,Q_h\HH)$.
  \item We have
    $\widetilde K\in \sL\Big(\{G_2\in E^\nu_0(C^1([-1,1])): \eqref{eq: necessary cdn G2}\text{ holds}\}\,,\,E^\nu_0(Q_h \DD)\Big)$
    and
    \begin{equation*}
      ||\widetilde K||_\nu\leq \widetilde\gamma(\nu)
    \end{equation*}
    for some continuous function $\widetilde \gamma:[0,p_0)\rightarrow[0,\infty)$.
  \end{enumerate}
\end{theorem}
We shall prove this theorem at the end of this appendix.  First we state a lemma, the proof of which is elementary and hence
omitted.

\begin{lemma}
  \label{lemma: widetilde G}
  Let $G_2\in E^{\nu}_0(C([-1,1]))$ with $\nu\in(-p_0,p_0)$.  For each $\v\in[-1,1]$, let the function $\widetilde G_2(\cdot,\v)$
  be defined as follows.  Let $\kappa\in C^\infty_0(\RR,[0,\infty))$ be such that $\int_\RR\kappa(t)dt=1$ and set
  \begin{align*}
    \widetilde G_2(t,\v)&=\cosh(p_0t)\Big(\int_{-\infty}^tG_2(u,\v)/\cosh(p_0u)\,du
    \\& \qquad{}-\int_{-\infty}^t\kappa(u)du\int_\RR  G_2(u,\v)/\cosh(p_0u)\,du\Big)
    \\&=\cosh(p_0t)\Big(-\int_{t}^{\infty} G_2(u,\v)/\cosh(p_0u)\,du
    \\& \qquad{}+\int_{t}^{\infty}\kappa(u)du\int_\RR  G_2(u,\v)/\cosh(p_0u)\,du\Big).
  \end{align*}
  Then
  \begin{enumerate}
  \item $\widetilde G_2\in E^\nu_1(C([-1,1]))$,
  \item 
    $\displaystyle
    \partial_t\widetilde G_2(t,\v)=
    G_2(t,\v)
     $\\\mbox{}\hfill$\displaystyle
    -\cosh(p_0t)\kappa(t)\int_\RR  G_2(u,\v)/\cosh(p_0u)\,du
    +p_0\tanh(p_0t)\widetilde G_2(t,\v),
    $
  \item
    $G_2-\partial_t \widetilde G_2\in E^\nu_1(C([-1,1]))$,
  \item
    $\displaystyle
    \widetilde{\widetilde G}_2(t):=
    \partial_t\Big( G_2(t)-\partial_t\widetilde G_2(t,\v)\Big)
    \\=
    \{p_0\sinh(p_0t)\kappa(t)+\cosh(p_0t)\kappa'(t)\}\int_\RR  G_2(u,\v)/\cosh(p_0u)\,d\v
    \\\qquad{}-(p_0/\cosh(p_0t))^2\widetilde G_2(t,\v)
    -p_0\tanh(p_0t) \partial_t \widetilde G_2(t,\v)
    \in E^\nu_0(C([-1,1])),
    $
  \item the linear maps 
    $G_2\ni E^\nu_0(C([-1,1]))\rightarrow \widetilde G_2\in E^{\nu}_1(C([-1,1]))$
    and
    $G_2\ni E^\nu_0(C([-1,1]))\rightarrow \widetilde{\widetilde G}_2\in E^{\nu}_0(C([-1,1]))$
    are bounded.
  \end{enumerate}
\end{lemma}

Let us consider the last component of equation~\eqref{eq: equation for U}.  
\begin{proposition}
  \label{prop: exists and unique}
  Given $G\in E_0^\nu(Q_h \DD)$, let $U=(\z,\xx,\Z) \in E_0^{\nu}(\DD)\cap C^1(\RR,\HH)$ be a solution to~\eqref{eq:
    equation for U}. Then $\z\in E^\nu_1(\RR)$ and $\Z$ solves the equation
  \begin{equation*}
    \partial_t \Z=\partial_\v\Z+G_2,~~\Z(t,0)=\z(t),
  \end{equation*}
  which has the unique solution $\Z\in E^{\nu}_0(C^1([-1,1]))\cap C^1(\RR,C([-1,1]))$ given by
  \begin{equation*}
    \Z(t,\v)=\z(t+\v)-\int_{t}^{t+\v}G_2(\si,t+\v-\si)d\si.
  \end{equation*}
  Moreover, this defines an affine map $G_2\rightarrow \Z$ such that
  \begin{equation}
    \label{eq: norm Z}
    ||\Z||_{E^{\nu}_0(C^1([-1,1]))} \leq\Const
    \left(||\z||_{E^\nu_1(\RR)}
    +||G_2||_{E^{\nu}_0(C^1([-1,1]))}  \right).
  \end{equation}
\end{proposition}

\begin{proof}
  Clearly, the given function $\Z$ is a solution and the estimate holds for this $\Z$.  To check uniqueness, it is enough to
  consider the case $\z=0$ and $G_2=0$.  If $\Z$ is a solution, let $\widetilde \Z(t,\v)= \Z(t-\v,\v)$, that is,
  $\Z(t,\v)=\widetilde \Z(t+\v,\v)$.  Then $\Z$ and $\widetilde \Z$ are $C^1(\RR\times[-1,1])$, and
  $\partial_\v \widetilde \Z(t,\v)=0$ with $\widetilde \Z(t,0)=0$.  Hence $\widetilde \Z=0$.
\end{proof}

Thanks to Proposition~\ref{prop: exists and unique}, \eqref{eq: equation for U} becomes
\begin{align*}
    \partial_t\z&=\xx+G_0,\\
  \partial_t\xx
                &=\gamma\tau^2\Delta_D \z-\tau^2\z \\
                & \qquad{} -\gamma\tau^2\int_t^{t+1}G_2(\v,t+1-\v)d\v
                  -\gamma\tau^2\int_t^{t-1}G_2(\v,t-1-\v)d\v+G_1\\
    &=\gamma\tau^2\Delta_D \z-\tau^2\z-\gamma\tau^2\int_0^{1}G_2(t+1-\v,\v)d\v \\
                & \qquad{}
    -\gamma\tau^2\int_0^{-1}G_2(t-1-\v,\v)d\v+G_1.
\end{align*}
Thus, we need to find $\z\in E^{\nu}_1(\RR)$ such that $\partial_t\z-G_0\in C^1(\RR)$ and solving
\begin{multline}
  \partial_t (\partial_t\z-G_0)
  =\gamma\tau^2\Delta_D \z-\tau^2\z-\gamma\tau^2\int_0^{1}G_2(t+1-\v,\v)d\v \\
  -\gamma\tau^2\int_0^{-1}G_2(t-1-\v,\v)d\v+G_1\,.
  \label{eq: second equivalent}
\end{multline}
If in addition $G_2\in E^\nu_1(C[-1,1])$, \eqref{eq: G_0} implies $G_0\in E^\nu_1(\RR)$ and the equation reads (for
$u\in E_1^\nu(\RR)\cap C^2(\RR)$ now)
\begin{align}
  L\z&:=
       \gamma^{-1}\tau^{-2}\z''-\Delta_D \z+\gamma^{-1}\z
       \notag\\
     & =
       -\int_0^{1}G_2(t+1-\v,\v)d\v
       -\int_0^{-1}G_2(t-1-\v,\v)d\v \notag \\
     & \qquad{}+\gamma^{-1}\tau^{-2}G_1
       +\gamma^{-1}\tau^{-2}G_0'
       \notag\\&\stackrel{
                 \eqref{eq: G_0},
                 \eqref{eq: G_1}
                 }=-\int_0^1G_2(t+1-\v,\v)d\v+\int_0^1G_2(t-1+\v,-\v)d\v
                 \notag\\&\qquad
                           +\int_0^1\cos(\k_0(1-\v))[G_2(t,\v)-G_2(t,-\v)]d\v
                           \notag\\&\qquad
                                     +\k_0^{-1}\int_0^1\sin(\k_0(1-\v))[\partial_tG_2(t,\v)+\partial_tG_2(t,-\v)]d\v
                                     =:Q(G_2).
                                     \label{eq: equation with G} 
\end{align}

\begin{proposition}
  \label{prop: sln of equation with G}
  If $\nu\in(-p_0,0)$ and $G_2\in E^\nu_1(C([-1,1]))$, then equation \eqref{eq: equation with G} has a solution
  $\z\in E^{\nu}_1(\RR) \cap C^2(\RR)$ such that
  \begin{equation*}
    ||\z||_{E^{\nu}_1(\RR)}\leq \Const||G_2||_{ E^\nu_0(C([-1,1]))}
  \end{equation*}
  uniformly in $\nu$ on compact subsets of $(-p_0,0)$.
\end{proposition}

\emph{Remarks.}  Observe that, in contrast with the previous results, it is assumed that $\nu$ is negative.  Moreover, we require
$G_2\in E^\nu_1(C([-1,1]))$ in the hypotheses and thus~\eqref{eq: equation with G} makes sense. However, in the conclusion, the
weaker norm $||\cdot||_{E^\nu_0(C([-1,1]))}$ is used.  As the norm $||G_2||_{E^\nu_0(C^1([-1,1]))}$ is needed in~\eqref{eq: norm
  Z} to control $||\Z||_{E^{\nu}_0(C^1([-1,1]))}$, in the end the norm in the statement of Theorem~\ref{thm: solution U} is
$||G||_{E^\nu_0(C^1([-1,1]))}$.

\begin{proof}
As $\nu\in(-p_0,0)$, we can consider the Fourier transform $\widehat G_2(k,\v)$ of $G_2(t,\v)$ with respect to $t$. The
Fourier transform $\sF[Q(G_2)]$ of the right-hand side of~\eqref{eq: equation with G} is
\begin{align*}
  & \int_0^1\Big(
    -e^{i\k(1-\v)}
  +\cos(\k_0(1-\v))
  +\k_0^{-1}\sin(\k_0(1-\v))i\k
  \Big)\widehat G_2(\k,\v)d\v
  \\ &\qquad{}+\int_0^1\Big(
  e^{i\k(-1+\v)}
  -\cos(\k_0(1-\v))
  +\k_0^{-1}\sin(\k_0(1-\v))i\k
  \Big)\widehat G_2(\k,-\v)d\v
  \\ &
  =\int_0^1\Big(
  \{\cos(\k_0(1-\v))-\cos(\k(1-\v))\}
  \\ &\qquad\qquad
  {}+i\{\k_0^{-1}\sin(\k_0(1-\v))\k-\sin(\k(1-\v))\}
  \Big)\widehat G_2(\k,\v)d\v
  \\&\qquad{}+\int_{-1}^0\Big(
  \{\cos(k(1+\v))-\cos(\k_0(1+\v))\}
  \\ &\qquad\qquad
  +i\{\k_0^{-1}\sin(\k_0(1+\v))\k-\sin(\k(1+\v))\}
  \Big)\widehat G_2(\k,\v)d\v
  \\&
  =\int_{-1}^1\Big(
  \text{sgn}(\v)\{\cos(\k_0(1-|\v|))-\cos(\k(1-|\v|))\}\widehat G_2(\k,\v)
  \\ &\qquad\qquad{}+\{\sinc(\k_0(1-|\v|))-\sinc(\k(1-|\v|))\}
  (1-|\v|)i\k\widehat G_2(\k,\v)\Big)d\v,
\end{align*}
where $\text{sgn}(0)=0$ and $\sinc$ is the cardinal sine function, i.e., $\sinc(\k)=\sin(\k)/\k$ ($=1$ at $\k=0$).

Let $\widetilde G_2(\cdot,\v)$ and $\widetilde{\widetilde G}_2$ be as in Lemma~\ref{lemma: widetilde G}.  As
$G_2\in E^\nu_1(C([-1,1]))$, $\widetilde G_2\in E^\nu_2(C([-1,1]))$. Because of
$\partial_tG_2=\partial_t^2\widetilde G_2+\widetilde{\widetilde G}_2$, we have for $\nu\in (-p_0,0)$
\begin{equation*}
  i\k\widehat G_2=-\k^2\sF[\widetilde G_2]+\sF[\widetilde{\widetilde G}_2]\,,
\end{equation*}
where the Fourier transforms are taken with respect to the first variable only.  Hence
\begin{multline*}
  \sF[Q(G_2)]=
  \int_{-1}^1\Big(
  \text{sgn}(\v)\{\cos(\k_0(1-|\v|))-\cos(\k(1-|\v|))\}\widehat G_2(\k,\v)+\\
  \{\sinc(\k_0(1-|\v|))-\sinc(\k(1-|\v|))\}
  (1-|\v|)\{-\k^2\sF[\widetilde G_2](\k,\v)
  +\sF[\widetilde{\widetilde{G}}_2](\k,\v)\}\Big)d\v.
\end{multline*}
Define
\begin{equation}
  \label{eq: def G_3}
  \widehat G_3(\k):=
  \int_{-1}^1\{\sinc(\k_0(1-|\v|))-\sinc(\k(1-|\v|))\}
  (1-|\v|)\sF[\widetilde G_2](\k,\v)d\v.
\end{equation}
Clearly, $\widehat G_3(\pm\k_0)=0$.
At the end of the proof, we shall check that $G_3\in E_2^\nu(\RR)$.  

To analyse the left-hand side of~\eqref{eq: equation with G}, we consider
\begin{equation*}
L(\z-\gamma\tau^2G_3)=
  \gamma^{-1}\tau^{-2}(\z-\gamma\tau^2G_3)''
  -\Delta_D (\z-\gamma\tau^2G_3)+\gamma^{-1}(\z-\gamma\tau^2G_3),
\end{equation*}
whose Fourier transform, using~\eqref{eq: def G_3}, equals to
\begin{multline*}
  \sF[L(\z-\gamma\tau^2G_3)]=\sF[Q(G_2)]-\gamma\tau^2 D(\k)\widehat G_3(\k)
  =
  \gamma\tau^2(2\cos(\k)-2 -\gamma^{-1})\widehat G_3
  \\
  +\int_{-1}^1\Big(
  \text{sgn}(\v)\{\cos(\k_0(1-|\v|))-\cos(\k(1-|\v|))\}\widehat G_2(\k,\v)
  \\+\{\sinc(\k_0(1-|\v|))-\sinc(\k(1-|\v|))\}
  (1-|\v|)\sF(\widetilde{\widetilde{G}}_2)(k,\v)\}\Big)d\v.
\end{multline*}
Note that by construction, the Fourier transform above vanishes at $\k=\pm\k_0$, the only real roots of the dispersion function
$D$.  Hence, by~\eqref{eq: def G_3},
\begin{multline}
  \sF[\z-\gamma\tau^2G_3](\k)=\int_{-1}^1 \widehat H_1(\k,\v)\widehat G_2(\k,\v)d\v
  +\int_{-1}^1 \widehat H_2(\k,\v)\sF(\widetilde{\widetilde G}_2)(\k,\v)d\v
  \\+\int_{-1}^1 \widehat H_3(\k,\v)\sF(\widetilde G_2)(\k,\v)d\v,
\end{multline}
where $\widehat H_j(\k,\v)$ is continuous in $(\k,\v)$ for $\v\neq 0$, the function $\k\rightarrow \widehat H_j(\k,\v)$ is
analytic in the strip $\{\k\in \CC:\Im\, \k \in(-p_0,p_0)\}$ and $(1+|\k|^2)H_j(\k,\v)$ is bounded in
$\{(\k,\v)\in \CC\times[-1,1]:\Im\, \k \in[-\delta,\delta]\}$ for $j=1,2,3$ and all $\delta\in(0,p_0)$.  For example,
\begin{equation*}
  \widehat H_1(\k,\v)=\frac{
  \text{sgn}(\v)\{\cos(\k_0(1-|\v|))-\cos(\k(1-|\v|))\}}{D(\k)}\,,
\end{equation*}
$k_0$ and $-k_0$ being removable singularities for $s\neq 0$. Again by Proposition~\ref{prop: adaptation}, the map that sends
$G_2\mapsto (G_2,\widetilde{\widetilde G}_2,\widetilde G_2) \mapsto \z-\gamma\tau^2 G_3\in E^\nu_1(\RR)\cap H^2(\RR)$ is well
defined and bounded from $\Big(E^{\nu}_{0}(C([-1,1]))\Big)^3$ to $E^{\nu}_{1}(\RR)$,
\begin{multline}
  \label{eq: eq for z}
  \z-\gamma\tau^2G_3=\int_{-1}^1  H_1(\cdot,\v)\star G_2(\cdot,\v)d\v
  +\int_{-1}^1  H_2(\cdot,\v)\star\widetilde{\widetilde G}_2(\cdot,\v)d\v
  \\+\int_{-1}^1  H_3(\cdot,\v)\star\widetilde G_2(\cdot,\v)d\v,
\end{multline}
where the convolutions are with respect to the first variable only.  Moreover $\z$ is a solution to~\eqref{eq: equation with
  G}. By Proposition~\ref{prop: adaptation},
\begin{multline*}
  ||\z-\gamma\tau^2G_3||_{E^{\nu}_1(\RR)} \\
  \leq \Const\left(||G_2||_{E^{\nu}_0(C([-1,1]))}
  +||\widetilde {\widetilde G}_2||_{E^{\nu}_0(C([-1,1]))}
  +||\widetilde G_2||_{E^{\nu}_0(C([-1,1]))}\right)
\end{multline*}
and
\begin{equation*}
  ||G_3||_{E^{\nu}_1(\RR)}
  \leq \Const_1||\widetilde G_2||_{E^{\nu}_1(C([-1,1]))}
  \leq \Const_2|| G_2||_{E^{\nu}_0(C([-1,1]))}\,
\end{equation*}
uniformly in $\nu$ in compact subsets of $(-p_0,0)$.  To see that $\Const_1$ is finite, rewrite \eqref{eq: def G_3} once more
\begin{multline*}
  \widehat G_3(\k)=
  \int_{-1}^1\sinc(\k_0(1-|\v|))
  (1-|\v|)\sF[\widetilde G_2](\k,\v)d\v
  \\-\int_{-1}^1\sinc(\k(1-|\v|))
  (1-|\v|)\sF[\widetilde G_2](\k,\v)d\v,
\end{multline*}
and observe that 
$|\sinc(\k(1-|\v|))(1-|\v|)|\leq \Const |\k|^{-1}$ on $\{(\k,\v)\in \CC\times[-1,1]: \Im\, \k\in(-p_0,p_0)\}$.  Then
Proposition~\ref{prop: adaptation} allows to transform back
\begin{multline}
  \label{eq: G_3 as G_2}
  G_3=
  \int_{-1}^1\sinc(\k_0(1-|\v|))
  (1-|\v|)\widetilde G_2(\cdot,\v)d\v
  \\-\int_{-1}^1\sF^{-1}\{\sinc(\cdot(1-|\v|))
  (1-|\v|)\}\star\widetilde G_2(\cdot,\v)d\v,
\end{multline}
with $ ||G_3||_{E^{\nu}_1(\RR)} \leq \Const_1||\widetilde G_2||_{E^{\nu}_1(C([-1,1]))}$ for a finite constant $\Const_1$.
Moreover $G_3\in E_2^\nu(\RR)$ as $\widetilde G_2\in E_2^\nu(C([-1,1]))$.
\end{proof}

\subsection*{Proof of Theorem~\protect{\ref{thm: solution U}}}
\label{sec:Proof-Theor-solut}

Proposition~\ref{prop: sln of equation with G} ensures the existence of a solution $\z\in E^\nu_1(\RR)$ to equations~\eqref{eq:
  G_0}, \eqref{eq: G_1}, \eqref{eq: second equivalent} for $G_2\in E^\nu_1(C([-1,1]))$ and $\nu\in (-p_0,0)$. However these
equations also make sense for $G_2\in E^\nu_0(C([-1,1]))$.  By an approximation procedure, the existence of a solution $\z$
to~\eqref{eq: G_0}, \eqref{eq: G_1}, \eqref{eq: second equivalent} and the estimate
\begin{equation}
  \label{eq: norm z}
  ||\z||_{E^{\nu}_1(\RR)}\leq \Const||G_2||_{ E^\nu_0(C([-1,1]))}
\end{equation}
of Proposition~\ref{prop: sln of equation with G} remain true for all $G_2\in E^{\nu}_0(C([-1,1]))$ (uniformly in $\nu$ in compact
subsets of $(-p_0,0)$). The approximation procedure thus defines a bounded linear map $G_2\mapsto \z$.  More explicitly, the
approximation can be done as follows: take $G_2\in E^{\nu}_0(C([-1,1]))$ and approximate it by a sequence
$\{G_{2,n}\}\subset E^{\nu}_1(C([-1,1]))$ that converges to $G_2$ in the Banach space $E^{\widetilde\nu}_0(C([-1,1]))$ for some
fixed $\widetilde \nu\in(\nu,0)$.  We then get a solution $\z_{\widetilde \nu}\in E^{\widetilde \nu}_1(C([-1,1]))$ such that
$||\z_{\widetilde\nu}||_{E^{\widetilde\nu}_1(\RR)}\leq \Const||G_2||_{ E^{\widetilde\nu}_0(C([-1,1]))}$.  For
$\widetilde{\widetilde \nu}\neq \widetilde \nu$, the equation $L(\z_{\widetilde{\widetilde \nu}}-\z_{\widetilde \nu})=0$ implies
that $\z_{\widetilde{\widetilde\nu}}=\z_{\widetilde \nu}$, by Proposition~\ref{prop: on the exponential decay}.  Hence we can
write $\z=\z_{\widetilde\nu}\in \cap_{s\in(\nu,0)}E^{s}_1(C([-1,1]))$. Thanks to the fact that the estimates are uniform in $\nu$
on compact subsets in $(-p_0,0)$, we get that $\z\in E^{\nu}_1(C([-1,1]))$ satisfies~\eqref{eq: norm z}.

This linear map is well-defined also when $\nu\in[0,p_0)$, the constants being in fact uniform in $\nu$ in every compact subset of
$(-p_0,p_0)$ (see~\eqref{eq: eq for z}, \eqref{eq: G_3 as G_2}, Proposition~\ref{prop: adaptation} and Lemma~\ref{lemma: widetilde
  G}) but it must be checked that $\z$ also gives rise to a solution when $\nu\in[0,p_0)$.  This can be done by a truncation that
brings the case $\nu\in[0,p_0)$ back to the case $\nu\in(-p_0,0)$.  Namely let $\nu\in[0,p_0)$, $\nu_+=(\nu+p_0)/2$ and
$\nu_-=(-p_0-\nu)/2$. Let $\zeta\in C_0^\infty(\RR,[0,\infty))$ be equal to $1$ in a neighbourhood of $0$.  Then, for
$G_2\in E_0^{\nu}(C([-1,1]))$, the sequence $\{\zeta(\cdot/n)G_2\}_{n\geq 1}\subset E_0^{\nu_-}(C([-1,1]))$ converges to $G_2$ in
$E^{\nu_+}_0(C([-1,1]))$ and is bounded in $E_0^{\nu}(C([-1,1]))$. Hence it is a Cauchy sequence in $E^{\nu_+}_0(C([-1,1]))$.  The
corresponding sequence $\{\z_n\}_{n\geq 1}\subset E^{\nu_-}_1(\RR)$ therefore converges in $E^{\nu_+}_1(\RR)$ to some
$\z \in E^{\nu}_1(\RR)$.  As each $\z_n$ solves \eqref{eq: G_0}, \eqref{eq: G_1}, \eqref{eq: second equivalent} with the
right-hand sides defined from $\zeta(\cdot/n)G_2$, it follows that $\z$ solves \eqref{eq: G_0}, \eqref{eq: G_1}, \eqref{eq: second
  equivalent} with the right-hand sides defined from $G_2$, giving rise in this way to a bounded linear map $G_2\mapsto \z$.  This
proves the first part of Theorem~\ref{thm: solution U}.

Let us prove the second part of Theorem~\ref{thm: solution U}.  Firstly, assume that $\nu\in(-p_0,0)$.  If $G\in E_0^\nu(Q_h\DD)$,
then~\eqref{eq: equation for U} gives $\partial_t P_1U=L_{\gamma,\tau}P_1U$.  As a consequence $U(t)\not\in Q_h\DD$ for some
$t\in \RR$ would imply that $P_1U$ is a non-trivial periodic solution on the centre manifold, in contradiction with
$\lim_{|t|\rightarrow \infty}||U(t)||_{\DD}=0$ (as $\nu\in(-p_0,0)$).  See the paragraph containing~\eqref{eq:ua} before
Proposition~\ref{prop: G} for the fact that the centre manifold (here for the linear problem) is filled by the equilibrium and
periodic solutions.  The above truncation procedure allows one to conclude that $U\in E^\nu_0(Q_h\DD)$ also when $\nu\in[0,p_0)$.

Finally, we turn to the third part of Theorem~\ref{thm: solution U} about uniqueness.  Let us first study the special case $G=0$
for $\nu\in (-p_0,p_0)$.  As $\Z$ is uniquely determined by $\z$ (see Proposition~\ref{prop: exists and unique}), let us consider
any solution $\z$ in $E^{\nu}_1(\RR)\cap C^2(\RR)$ to
\begin{equation*}
  L\z=\gamma^{-1}\tau^{-2}\z''-\Delta_D \z+\gamma^{-1}\z=0.
\end{equation*}
Observe that $\z''\in E^\nu_0(\RR)$ and let $\nu_-\in (-p_0,-|\nu|)$. For all test functions
$\tilde \z\in H^2(\RR)\cap E_1^{\nu_-}(\RR)$, integrations by parts give
\begin{equation*}
0=\int_\RR L\z \cdot\tilde \z dt=
  \int_\RR \z L\tilde \z dt
.
\end{equation*}
By Proposition~\ref{prop: on the exponential decay}, the map $\tilde \z\mapsto L\tilde \z=:Q$ is surjective from
$H^2(\RR)\cap E_1^{\nu_-}(\RR)$ to
\begin{equation*}
  \left\{Q\in E_0^{\nu_-}(\RR):\int_\RR Q(t)\cos(\k_0t)dt
  =\int_\RR Q(t)\sin(\k_0t)dt=0\right\}.
\end{equation*}
Thus, for a solution $\z$ we get $\int_\RR \z(t)Q(t)dt=0$ for all such $Q$.

Let $\eta_c,\eta_s\in E_0^{\nu_-} (\RR)$ be such that
\begin{equation*}
  \int_{\RR}\eta_c(t)\cos(\k_0t)dt=\int_{\RR}\eta_s(t)\sin(\k_0t)dt=1,
\end{equation*}
\begin{equation*}
  \int_{\RR}\eta_c(t)\sin(\k_0t)dt=\int_{\RR}\eta_s(t)\cos(\k_0t)dt=0.
\end{equation*}
If $Q\in E_0^{\nu_-}(\RR)$, then
\begin{equation*}
  \widetilde Q:=Q-\eta_c\int_\RR Q(\u)\cos(\k_0\u)d\u
  -\eta_s\int_\RR Q(\u)\sin(\k_0\u)d\u
\end{equation*}
satisfies $\int_\RR \widetilde Q(t)\cos(\k_0t)dt=\int_\RR \widetilde Q(t)\sin(\k_0t)dt=0$ and therefore by Fubini
\begin{multline*}
  0=\int_\RR \z(t)\left(Q(t)-\eta_c(t)\int_\RR Q(\u)\cos(\k_0\u)d\u
  -\eta_s(t)\int_\RR Q(\u)\sin(\k_0\u)d\u\right)dt
  \\=\int_\RR\left( \z(t)-\cos(\k_0t)\int_\RR \z(\u)\eta_c(\u)d\u
  -\sin(\k_0t)\int_\RR \z(\u)\eta_s(\u)d\u\right)Q(t)dt.
\end{multline*}
Hence the function $\z$ is in the span of the two functions $\cos(\k_0\cdot)$ and $\sin(\k_0\cdot)$:
\begin{align*}
  \z&=\cos(\k_0\cdot)\int_\RR \z(\u)\eta_c(\u)d\u
  +\sin(\k_0\cdot)\int_\RR \z(\u)\eta_s(\u)d\u 
\\&\in 
  \text{span}\{\cos(\k_0\cdot),\sin(\k_0\cdot)\}
\end{align*}
and, for all $t\in \RR$, we have (see~\eqref{eq:ua} or Lemma 2 in~\cite{Iooss2000a})
\begin{multline*}
(\z(t),\z'(t),\z(t+\cdot))\in
  \text{span}\Big\{\Big(\cos(\k_0t),-\k_0\sin(t),\cos(\k_0(t+\cdot))\Big)\, ,\\
  \Big(\sin(\k_0t),\k_0\cos(t),\sin(\k_0(t+\cdot))\Big)\Big\}
  =P_1\DD,
\end{multline*}
where $\z(t+\cdot)$ denotes the function $\v\rightarrow \z(t+\v)$ for $\v\in[-1,1]$.

We are now ready to check the uniqueness of the solution $U\in E^{\nu}_0(Q_h\DD)\cap C^1(\RR,Q_h\HH)$ for $\nu\in(-p_0,p_0)$. It
is clearly sufficient to check it for $G=0$ only.  Moreover, as $\Z$ is unique for a unique $\z\in E^\nu_1(\RR)$ (see
Prop.~\ref{prop: exists and unique}), it is enough to show that $\z=0$ is the unique solution in $E^{\nu}_1(\RR)\cap C^2(\RR)$ to
the equation
\begin{equation*}
  L\z
  =0
  ~\text{ such that }~
  \forall t\in \RR~~U(t,\cdot):=(\z(t),\z'(t),\z(t+\cdot))\in Q_h\DD,
\end{equation*}
where $U(t,\cdot)$ denotes the function $\v\rightarrow U(t,\v)$ for $\v\in[-1,1]$.  As $L\z=0$, we have just seen that necessarily
$U(t)
\in P_1\DD
$ for all $t\in \RR$. Hence $\z=0$ as desired.

The last part of Theorem~\ref{thm: solution U} results from~\eqref{eq: norm Z} and~\eqref{eq: norm z}, where the various constants
are uniform in $\nu$ on any compact subset of $(-p_0,p_0)$.

\emph{Remark.}  A look into the proofs of the present appendix shows that the arguments work as well for
\begin{multline*}
  Y=\bigl\{G=(G_0,G_1,G_2):\, G_0,G_1 \in E_0^{\nu}(\RR),\, G_2\in E_0^{\nu}(C^1([-1,1])),\, \bigr. \\
  \bigl. \eqref{eq: G_0}\text{ and }\eqref{eq: G_1}\text{ hold}\bigr\}.
\end{multline*}

\paragraph{Acknowledgement} JZ gratefully acknowledges funding by the EPSRC through project EP/K027743/1, the Leverhulme Trust
(RPG-2013-261) and a Royal Society Wolfson Research Merit Award.


\makeatletter
\renewcommand\section{\@startsection {section}{1}{\z@}%
           {18\p@ \@plus 6\p@ \@minus 3\p@}%
           {9\p@ \@plus 6\p@ \@minus 3\p@}%
           {\normalsize\bfseries\boldmath \noindent}}
\makeatother
\bibliography{jz}

\end{document}